%% file: main.tex
\documentclass[runningheads]{llncs}
\usepackage[T1]{fontenc}
\usepackage{graphicx}
\usepackage{sty/_shortcuts}
\usepackage{quiver}
\usepackage[australian]{babel}
\usepackage[autostyle]{csquotes}
\usepackage{lmodern}
\usepackage{dutchcal}
\usepackage{microtype}
\usepackage{upgreek}
\usepackage[capitalize]{cleveref}
\usepackage[ddmmyy]{datetime}

\usepackage{stackengine}

\newif\ifarxiv{}

\newif\ifdraft{}

\arxivtrue{} % comment these out for publication verison

\ifdraft\usepackage[anchor=mr, hpos=205mm, fontsize=24pt, angle=90]{draftwatermark}
\SetWatermarkText{\textit{--- DRAFT --- \today}}
\else\renewcommand*{\todo}[1]{}
\fi

\newenvironment{notation}{\paragraph{Notation.}}{} % notation env

\begin{document}

\title{Trace models of concurrent valuation algebras
    \ifdraft{}\\~\textcolor{red}{\small{Draft \today: please do not distribute}}\fi
}
\author{Naso Evangelou-Oost\orcidID{0000-0002-8313-6127} \and
    Larissa Meinicke\orcidID{0000-0002-5272-820X} \and
    Callum Bannister\orcidID{0000-0002-8799-054X} \and
    Ian J. Hayes\orcidID{0000-0003-3649-392X}}
\authorrunning{N. Evangelou-Oost et al.}
\institute{The University of Queensland, St Lucia, Australia}

\maketitle
\begin{abstract}
    \input{content/abstract}
\end{abstract}

%% We should not differentiate between these notations
\renewcommand*{\projtup}{\proj}

\newcommand*{\nattop}{\boldsymbol{\top}}
\newcommand*{\natepsilon}{\boldsymbol{\upepsilon}}
\newcommand*{\natiota}{{\boldsymbol{\upiota}}}
\newcommand*{\nattau}{{\boldsymbol{\uptau}}}

%% ordinary \dot is too high on glue and conc
\newcommand\lowdot[1]{\stackon[-2pt]{$#1$}{.}}

\newcommand*{\shuf}{\mathbin{\shuffle}}
\newcommand*{\shufdom}[1]{\mathbin{\shuffle_{#1}}}
\newcommand*{\shuftrace}[1]{\mathbin{\dot{\shuffle}_{#1}}}

\newcommand*{\conc}{\mathbin{\smallfrown}}
\newcommand*{\concdom}[1]{\mathbin{\smallfrown_{#1}}}
\newcommand*{\conctrace}[1]{\mathbin{\lowdot{\smallfrown}}_{#1}}

\newcommand*{\glue}{\mathbin{\smallsmile}}
\newcommand*{\gluedom}[1]{\mathbin{\smallsmile_{#1}}}
\newcommand*{\gluetrace}[1]{\mathbin{\lowdot{\smallsmile}}_{#1}}

\newcommand*{\rglue}{\mathbin{\rel{\smallsmile}}}
\newcommand*{\rgluedom}[1]{\mathbin{\rel{\smallsmile_{#1}}}}
\newcommand*{\rgluetrace}[1]{\mathbin{\cdot_{#1}}}

\newcommand*{\rnatjoin}{\mathbin{\rel{\natjoin}}}

\newcommand*{\groth}[1]{\smallint{#1}}

\newcommand*{\gphi}{\groth{\fphi}}
\newcommand*{\gpsi}{\groth{\fpsi}}

\newcommand*{\fpsi}{\fun{\Psi}}
\newcommand*{\fphi}{\fun{\Phi}}
\newcommand*{\fomega}{\fun{\Omega}}

% since natjoin = meet, just use this symbol.
\newcommand*{\natjoin}{\wedge}
\newcommand*{\emptystate}{\heartsuit}

\newcommand*{\aomega}{\fomega^{\mathsf{act}}}
\newcommand*{\somega}{\fomega^{\mathsf{state}}}

\newcommand*{\atup}{\fun{T}^{\mathsf{act}}}
\newcommand*{\stup}{\fun{T}^{\mathsf{state}}}

\newcommand*{\fstate}{\fun{\Sigma}}
\newcommand*{\faction}{\fun{\Gamma}}

\newcommand*{\gstate}{\groth{\fstate}}
\newcommand*{\gstaterel}{\groth{\fstaterel}}
\newcommand*{\gaction}{\groth{\faction}}

% since skip is Top, just use the top symbol
\newcommand*{\ski}{\nattop}
\newcommand*{\nil}{\nattau}
\newcommand*{\emp}{\natiota}

\newcommand*{\neutralseq}{\nat{\code{skip}}}
\newcommand*{\neutralpar}{\nat{\code{run}}}

\newcommand*{\hoare}[3]{#1 \; \set{#2} \; #3}
\newcommand*{\jones}[5]{#1 \; #2 \; \set{#3} \; #4 \; #5}

\newcommand*{\lists}{\fun{L}}
\newcommand*{\listsi}{\fun{L_{+}}}

\newcommand*{\rel}[1]{#1^\mathsf{rel}}
\newcommand*{\fstaterel}{\rel{\fstate}}

\newcommand*{\rtup}{\rel{\fun{T}}}

\newcommand*{\rnil}{\rel{\nil}}
\newcommand*{\rski}{\rel{\ski}}

\newcommand*{\bow}{\mathbin{\bowtie}}

% Simpler proj/extension notation (Bad idea: then we can't use subscripts on valuations)
% \renewcommand*{\ext}[2]{{#1}^{#2}}
% \renewcommand*{\proj}[2]{{#1}_{#2}}

\input{content/intro}
\input{content/ova}
\input{content/cva}
\input{content/tuple_systems}
\input{content/action_trace}
\input{content/state_trace}
\input{content/relative_trace}
\input{content/computation}
\input{content/discussion}
\input{content/ack}
\ifarxiv\input{content/appendix}\fi

\bibliographystyle{splncs04}
\bibliography{main.bib}

\end{document}

%% file: content/abstract.tex
% The abstract should briefly summarize the contents of the paper in 150--250 words.

This paper introduces Concurrent Valuation Algebras (CVAs), a novel extension of ordered valuation algebras (OVAs). CVAs include two combine operators representing parallel and sequential products, adhering to a weak exchange law. This development offers theoretical and practical benefits for the specification and modelling of concurrent and distributed systems.
As a presheaf on a space of domains, CVAs enable localised specifications, supporting modularity, compositionality, and the ability to represent large and complex systems. Furthermore, CVAs align with lattice-based refinement reasoning and are compatible with established methodologies such as Hoare and Rely-Guarantee logics.
The flexibility of CVAs is explored through three trace models, illustrating distinct paradigms of concurrent/distributed computing, interrelated by morphisms. The paper also highlights the potential to incorporate a powerful local computation framework from valuation algebras for model checking in concurrent and distributed systems.
The foundational results presented have been verified with the proof assistant Isabelle/HOL.\@

% \keywords{First keyword \and Second keyword \and Another keyword.}

\keywords{Concurrent valuation algebras \and Concurrent systems \and Distributed systems.}

%% file: content/intro.tex
\section{Introduction}
% everyone else thinks

Valuation algebras are versatile algebraic structures that parameterise information across multiple domains, representing for example subsets of variables or events. These structures have been widely utilised across diverse disciplines such as database theory, logic, probability and statistics, and constraint satisfaction, among others. What sets valuation algebras apart is their robust computational theory, enabling the deployment of highly efficient distributed algorithms for addressing inference problems that involve information combination and querying~\cite{pouly2012generic}.

In our preceding work~\cite{Evangelou23}, we applied ordered valuation algebras to distributed systems, demonstrating their potential as a modular framework for specifying these systems in a refinement paradigm. Moreover, we established a link between sequential consistency---a crucial correctness criterion---and contextuality, an abstract form of information inconsistency which valuation algebras capture.

\subsubsection{Paper outline.}

In \cref{sec:ova}, we introduce ordered valuation algebras (OVAs) based on prior studies~\cite{DBLP:journals/ijar/Haenni04,MR4024458} and extend these to concurrent valuation algebras (CVAs) in \cref{sec:cva}, a structure comprising two OVA structures on a space with combine operators adhering to a weak exchange law. This design takes inspiration from \emph{Communicating Sequential Processes (CSP)}~\cite{DBLP:books/ph/Hoare85}, \emph{Concurrent Kleene Algebra (CKA)}~\cite{DBLP:journals/jlp/HoareMSW11}, \emph{Concurrent Refinement Algebra (CRA)}~\cite{DBLP:journals/fac/HayesMWC19}, and \emph{duoidal/$2$-monoidal categories}~\cite{MR2724388}. We also define morphisms between CVAs, and explain their alignment with the refinement reasoning methodologies of Hoare~\cite{DBLP:journals/cacm/Hoare69} and Rely-Guarantee~\cite{jones1981development} logics.
\Cref{sec:tuple} delves into tuple systems and relational OVAs, which underpin our trace models.
The subsequent sections, \cref{sec:action_trace_model,sec:state_trace_model,sec:relative_state_model}, introduce and contrast various trace models, elaborating on their distinct combine operators and trace characteristics.
In \cref{sec:computation}, we reflect on the potential extension of the local computation framework from valuation algebras to CVAs.
\cref{sec:discussion} closes our paper, encapsulating our findings and suggesting avenues for future exploration.

The theoretical underpinnings detailed in \cref{sec:ova,sec:cva,sec:tuple} have been rigorously formalised using the proof assistant Isabelle/HOL, lending credibility to our study.\footnote{\label{fn:isabelle}Available at \url{https://github.com/nasosev/cva} .} A separate formalisation of \cref{sec:state_trace_model} is also available.\footnote{Available at \url{https://github.com/onomatic/icfem23-proofs} .}
\ifarxiv\else{
        Proofs of omitted results can be found in the appendix of the arXiv version of this paper~\cite{evangelouoost2023trace}.
    }
\fi

%% file: content/ova.tex
\section{Ordered valuation algebras}
\label{sec:ova}

We assume familiarity with foundational ideas in order theory and category theory, including the definitions of a category, a functor, and a natural transformation. For those interested in a more detailed understanding, please refer to~\cite{BrendanFong2022SSiC} for an accessible introduction, or~\cite{riehl2017category} as a thorough reference.

\begin{notation}
    The category of sets and functions is denoted $\cat{Set}$.
    Posets (partially ordered sets) are identified with their associated (thin) categories, so that the hom-set $\hom(a,b)$ is a singleton when $a \leq b$ and empty otherwise, and we write $\cat{Pos}$ for the category whose objects are posets and whose morphisms are monotone functions.
    A \defn{topological space} $(X,\cat{T})$ is a set $X$ equipped with a topology $\cat{T}$, which is a family of subsets of $X$, termed \defn{open sets}, partially ordered by inclusion, and closed under arbitrary unions and finite intersections (in particular, the empty intersection, $X$, and the empty union, $\emptyset$, are open).
    For a category $\cat{C}$, $\opcat{C}$ denotes its \defn{opposite category}, that is the category $\cat{C}$ with the direction of its arrows reversed.
    A \defn{presheaf} $\fphi$ on a topological space $(X, \cat{T})$ is a functor with domain $\opcat{\cat{T}}$.
    We notate the value of a presheaf $\fphi$ applied to $A$ by $\fphi_A$ instead of $\fphi(A)$, and for $B \subseteq A$ in $\cat{T}$, the \defn{restriction map} $\fphi_A \to \fphi_B$ is denoted $a \mapsto \proj{a}{B}$.
    All presheaves considered are valued in $\cat{Set}$ or $\cat{Pos}$, though we adopt the name \defn{prealgebra} for a poset-valued presheaf, suggesting our intent to develop an OVA structure upon it.
    A \defn{global element} of a prealgebra $\fphi$ is a natural transformation $\natepsilon : \fun{1} \nattra \fphi$ from the terminal prealgebra $\fun{1}$, defined $\fun{1}\defeq A \mapsto \set{\emptystate}$.
    For such a global element $\natepsilon$, we write $\natepsilon_A$ instead of $\natepsilon_A(\emptystate)$.
    The symbol $\fun{P}$ denotes the \defn{covariant powerset functor} $\fun{P} : \cat{Set} \to \cat{Pos}$, sending a set $X$ to the poset of its subsets $\fun{P}(X)$, and sending a function $f : X \to Y$ to its \defn{direct image} $f_* \defeq X \mapsto \setc{ f(x) }{ x \in X }$.
    The symbol $\Nn$ denotes the set of natural numbers $\set{0,1,\ldots}$, while $\NnOne$ is the set of positive natural numbers $\set{1,2,\ldots}$.
\end{notation}

Throughout this paper, we fix a topological space $(X,\cat{T})$. Here, the open sets symbolise abstract \emph{domains}, representing subsets of system elements like memory locations, resources, or events, as well as their interconnectivity.

\begin{example}
    A network composed of three computer systems ${a, b, c}$ and three network links ${d,e,f}$ as pictured in \cref{fig:triangle} may be represented by the topological space generated by unions and intersections of the domains $\set{d,a,e}$, $\set{e,b,f}$, $\set{f,c,d}$.
    More generally, a network defined by a labelled, undirected graph converts to a finite topology where open sets are the upwards-closed sets of the network's \emph{face poset}, i.e.~the poset whose elements are the nodes $n$ and edges $e$ of the network, where $n \leq e$ if and only if $n$ is a vertex of $e$.\footnote{The topology described is the \emph{Alexandrov} topology of the face poset of the network, viewed as a simplicial complex. Another possibility is to take its geometric realisation, but this typically results in an infinite space.
        These spaces, however, are weakly homotopy equivalent~\cite{MR3024764}.}
    Alternatively, a set of memory addresses $X$ may be given the discrete topology $\cat{T} = \fun{P}(X)$.
    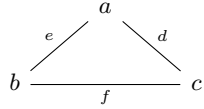
\begin{figure}
        % https://q.uiver.app/#q=WzAsMyxbMCwxLCJiIl0sWzEsMCwiYSJdLFsyLDEsImMiXSxbMCwxLCJlIiwwLHsic3R5bGUiOnsiaGVhZCI6eyJuYW1lIjoibm9uZSJ9fX1dLFswLDIsImYiLDIseyJzdHlsZSI6eyJoZWFkIjp7Im5hbWUiOiJub25lIn19fV0sWzIsMSwiZCIsMix7InN0eWxlIjp7ImhlYWQiOnsibmFtZSI6Im5vbmUifX19XV0=
        \[\begin{tikzcd}
                & a \\
                b && c
                \arrow["e", no head, from=2-1, to=1-2]
                \arrow["f"', no head, from=2-1, to=2-3]
                \arrow["d"', no head, from=2-3, to=1-2]
            \end{tikzcd}\]
        \caption{A network of three computers and three links.}
        \label{fig:triangle}
    \end{figure}
\end{example}

A prealgebra $\fphi : \opcat{T} \to \cat{Pos}$ comprises a family of posets $\set{\fphi_A}_{A \in \cat{T}}$ parameterised by the domains of the space $\cat{T}$, and a family of monotone restriction maps $\set{a \mapsto \proj{a}{B} : \fphi_A \to \fphi_B}_{A, B \in \cat{T}, B \subseteq A}$ parameterised by the inclusions of the space.

The elements of the posets $\fphi_A$ represent abstract units of information pertaining to their domain.
Their ordering signifies information refinement: $a \preceq b$ means $a$ is \emph{more} deterministic than $b$, a convention that aligns with the intuitions of program refinement.

The prealgebra's restriction maps $a\mapsto \proj{a}{B}$ serve to project or query information $a \in A$ onto a subdomain $B \subseteq A$.
These mappings facilitate the extraction of specific details from a wider context.
Further, restriction maps are transitive and idempotent: for $C \subseteq B \subseteq A$ and $a \in \fphi_A$, we have $\proj{(\proj{a}{B})}{C} = \proj{a}{C}$, and $\proj{a}{A} = a$.

This family of posets $\set{\fphi_A}_{A \in \cat{T}}$ can be unified into a single poset $\gphi$, through a canonical process known as the \emph{Grothendieck construction} of $\fphi$ (for a detailed explanation within a broader context, refer to~\cite{DBLP:journals/tcs/TarleckiBG91}).

% https://ncatlab.org/nlab/show/Grothendieck+construction
\begin{definition}[covariant Grothendieck construction for a prealgebra]
    \label{def:gc}
    Let $\fphi : \opcat{T} \to \cat{Pos}$ be a prealgebra.
    The \defn{covariant Grothendieck construction} of $\fphi$ is the poset $(\gphi, \preceq)$ whose elements are pairs $(A,a) \in \gphi$ where $A \in \cat{T}$ and $a \in \fphi_A$, and whose ordering $\preceq$ is defined
    \begin{equation}
        (A,a) \preceq (B,b)
        \text{ if and only if }
        B \subseteq A \text{ and }
        \proj{a}{B} \leq_{\fphi_B} b
    \end{equation}

    For the projection map $\dd : \gphi \to \opcat{T} , (A,a) \mapsto A$, call $\dd a$ the \defn{domain of $a$}.
\end{definition}

\begin{notation}
    As shorthand, we suppress the domain in the first component of elements $(A,a)$ belonging to $\gphi$, writing $a$ instead of $(A, a)$.
\end{notation}

\begin{remark}
    In \cref{def:gc}, we apply the \emph{covariant} Grothendieck construction to a contravariant functor, treating it as a covariant functor from its domain's opposite. This choice, though atypical, aligns with a semantic interpretation for refining program specifications, explained in \cref{sec:refinement}.
\end{remark}

Next, the concept of an \emph{ordered valuation algebra (OVA)} is introduced, which incorporates a prealgebra $\fphi : \opcat{T} \to \cat{Pos}$, a binary operator $\otimes : \gphi \times \gphi \to \gphi$, and a global element $\natepsilon : \fun{1} \nattra \fphi$, satisfying a number of axioms.
Before delving into the formal definition, we illustrate the concept with an example.

\begin{example}
    A familiar instance of an OVA models \emph{relational databases}. Here, a set $X$ of \emph{attributes} is fixed (e.g., $X = \{\text{{`name'}}, \text{{`age'}}, \text{{`height'}}\}$).
    A \emph{schema} is a subset $A$ of $X$ defining a table's columns, while each row defines a \emph{tuple}: an assignment of a value to each attribute. A \emph{relation} on $X$ is a set $a$ of tuples sharing a common schema $\dd a$, and a relational database is a set of such relations.

    To frame this within an OVA, we define a prealgebra $\fphi : \fun{P}(X) \to \cat{Pos}$, mapping a schema $A \in \cat{T}$ to the poset $\fphi_A$ of all relations with schema $A$, with ordering given by inclusion.
    The restriction maps of $\fphi$ correspond to \emph{querying}, by projecting the tuples of a relation $a$ to a sub-schema $B$; the result is the relation $\proj{a}{B} = \setc{ \proj{t}{B} }{t \in a}$, where $\proj{t}{B}$ is the tuple $t$ restricted to the attributes in $B$.

    The operator $\otimes$ is taken to be the \emph{natural join},
    \begin{equation}
        \begin{aligned}
            \bow     & : \gphi \times \gphi \to \gphi                                                                            \\
            a \bow b & \defeq \set{t \in \fphi_{\dd a \cup \dd b} \mid \proj{t}{\dd a} \in a \text{ and } \proj{t}{\dd b} \in b}
        \end{aligned}
    \end{equation}
    This operation is associative and monotone (forming an \emph{ordered semigroup}), and the schema of $a \bow b$ is $\dd a \cup \dd b$.
    Moreover, the natural join satisfies the following \emph{combination axiom}:
    \begin{equation}
        \proj{(a \bow b)}{\dd a} = a \bow \proj{b}{\dd a \cap \dd b}
    \end{equation}
    This identity is fundamental to query optimisation algorithms in relational databases, with its right-hand side referred to as a \emph{semi-join}.

    Lastly, the global element $\natepsilon$ assigns to each schema $A$ the universal relation $\natepsilon_A$ on $A$, encompassing all possible tuples on $A$.
    These universal relations serve as units for the natural join, i.e.~for all $a \in \gphi$, we have $a \bow \natepsilon_{\dd a} = a = \natepsilon_{\dd a} \bow a$.
\end{example}

Please note that our definition of an OVA below deviates from standard ones (e.g.~\cite{pouly2012generic,MR4024458,DBLP:journals/ijar/Haenni04}) in several ways.
First, we do not mandate commutativity of the operator $\otimes$, as a sequential product of programs is noncommutative.
This requires a symmetric revision~\cref{eq:combination} of the \emph{combination axiom}.
Second, constraints in the classical definition such as the existence of infima in the posets $\fphi_A$ are not imposed.
Yet, we stipulate that \emph{neutral} valuations correspond to a global element, which is tantamount to the \emph{stability} property in~\cite{pouly2012generic}, though we do not require neutral valuations to combine to neutral valuations---we call an algebra in which this property holds \emph{strongly neutral}~(\cref{defn:strongly_neutral}).
Last, while Grothendieck constructions have been applied to ordered valuation algebras~(\cite{DBLP:conf/wadt/ChenRT18}), their conventional definition does not involve a Grothendieck ordering.

\begin{definition}[ordered valuation algebra (OVA)]
    An \defn{ordered valuation algebra (OVA)} is a triple $(\fphi, \otimes, \natepsilon)$, where $\fphi$ is a prealgebra $\fphi : \opcat{T} \to \cat{Pos}$, $\otimes$ is a binary operator
    $
        \otimes : \gphi \times \gphi \to \gphi
    $,
    called the \defn{combine operator},
    and $\natepsilon : \fun{1} \nattra \fphi$ is a global element, called the \defn{neutral element}, satisfying the below four axioms for all \defn{valuations} $a ,b,c, a', b' \in \gphi$:
    \begin{description}
        \item[Ordered semigroup.] The combine operator $\otimes$ is associative, and monotone:
            \begin{equation}
                a \otimes (b \otimes c) = (a \otimes b) \otimes c \text{, and, }
                a \preceq a' \text{ and } b \preceq b' \implies a \otimes b \preceq a' \otimes b'
            \end{equation}
        \item[Labelling.]
            \begin{equation}
                \dd (a \otimes b) = \dd a \cup \dd b
            \end{equation}
        \item[Neutrality.]
            \begin{equation}
                \label{eq:neutrality}
                \natepsilon_{\dd a} \otimes a = a = a \otimes \natepsilon_{\dd a}
            \end{equation}
        \item[Combination.]
            \label{item:combination}
            \begin{equation}
                \label{eq:combination}
                \proj{(a \otimes b)}{\dd a} = a \otimes \proj{b}{\dd a \cap \dd b}
                ,\qquad
                \proj{(a \otimes b)}{\dd b} = \proj{a}{\dd a \cap \dd b} \otimes b
            \end{equation}
    \end{description}
\end{definition}

\begin{remark}
    The functor laws for $\fphi$ imply that for all $C \subseteq B \subseteq A$ and $a \in \fphi_A$, we have $\proj{(\proj{a}{B})}{C} = \proj{a}{C}$, and also $\proj{a}{A} = a$.
    The requirement that $\natepsilon$ is a global element says for all $B \subseteq A$ in $\cat{T}$, we have $\proj{\natepsilon_A}{B} = \natepsilon_B$.
\end{remark}

\begin{remark}
    \label{rem:local_mono}
    Monotonicity of $\otimes$ says for $a_1 \in \fphi_{A_1}$, $a_2 \in \fphi_{A_2}$, $b_1 \in \fphi_{B_1}$, $b_2 \in \fphi_{B_2}$,
    if $a_1 \preceq b_1$ and $a_2 \preceq b_2$,
    then
    $
        a_1 \otimes a_2 \preceq b_1 \otimes b_2
    $,
    that is,
    $
        \proj{(a_1 \otimes a_2)}{B_1 \cup B_2}
        \leq_{\fphi_{B_1 \cup B_2}}
        b_1 \otimes b_2
    $.
    Taking $A_1 = A_2 = B_1 = B_2$, this implies \defn{local monotonicity}, i.e.~the combine operator $\otimes$ restricted to each domain $A$, $\otimes_A : \fphi_A \times \fphi_A \to \fphi_A$, is monotone.
\end{remark}

\begin{definition}[commutative OVA]
    We call an OVA $(\fphi, \otimes, \natepsilon)$ \defn{commutative} if $\otimes$ is commutative.
\end{definition}

\begin{definition}[strongly neutral]
    \label{defn:strongly_neutral}
    We call an OVA $(\fphi, \otimes, \natepsilon)$ \defn{strongly neutral} if for all inclusions $B \subseteq A$ in $\cat{T}$, we have
    $
        \ext{\natepsilon_B}{A} = \natepsilon_A
    $.
\end{definition}

\begin{theorem}
    \label{thm:extension}
    Let $(\fphi, \otimes, \natepsilon)$ be an OVA.\@
    Then for each inclusion $B \subseteq A$ in $\cat{T}$, the restriction map $a \mapsto \proj{a}{B}$ has a right adjoint given by $b \mapsto \natepsilon_A \otimes b$.
    Moreover, these right adjoints assemble to a functor $\cat{T} \to \cat{Pos}$.
    We adopt the notation $\ext{b}{A} \defeq \natepsilon_A \otimes b$, and call $\ext{b}{A}$ the \defn{extension of $b$ to $A$}.
\end{theorem}
\begin{proof}
    We must show that for all $B \subseteq A$ and all $a \in \fphi_A$ and $b \in \fphi_B$,
    \begin{equation*}
        \proj{a}{B} \leq_{\fphi_B} b \iff a \leq_{\fphi_A} \natepsilon_A \otimes b
    \end{equation*}
    Assume $\proj{a}{B} \leq_{\fphi_B} b $.
    Using the fact that $a \preceq \proj{a}{B}$ and monotonicity,
    \begin{equation*}
        a
        = \natepsilon_A \otimes a
        \leq_{\fphi_A} \natepsilon_A \otimes \proj{a}{B}
        \leq_{\fphi_A} \natepsilon_A \otimes b
    \end{equation*}
    Now assume $a \leq_{\fphi_A} \natepsilon_A \otimes b$.
    Using monotonicity of restriction, the combination axiom, naturality of $\natepsilon$, and neutrality,
    \begin{equation*}
        \proj{a}{B}      \leq_{\fphi_B} \proj{( \natepsilon_A \otimes b)}{B}
        = \proj{\natepsilon_A}{A \cap B} \otimes b
        = \natepsilon_B \otimes b
        = b
    \end{equation*}
    So the adjunction holds.
    That extension is functorial (i.e.~for $C \subseteq B \subseteq A$ and $c \in \fphi_C$, both $\ext{(\ext{c}{B})}{A} = \ext{c}{A}$ and $\ext{c}{C} = c$) is due to the composability and uniqueness of adjoints.
\end{proof}

% Moreover, the converse holds for the first clause.
\begin{corollary}
    Let $(\fphi, \otimes, \natepsilon)$ be a strongly neutral OVA.\@
    Then for all $A,B$ in $\cat{T}$,
    $\natepsilon_A \otimes \natepsilon_B = \natepsilon_{A \cup B}$.
    Also, $\gphi$ is an ordered monoid.
\end{corollary}
\begin{proof}
    We have,
    $
        \natepsilon_A \otimes \natepsilon_B
        = (\natepsilon_{A \cup B} \otimes \natepsilon_{A \cup B}) \otimes (\natepsilon_A \otimes \natepsilon_B)
        = (\natepsilon_{A \cup B} \otimes \natepsilon_A) \otimes (\natepsilon_{A \cup B} \otimes \natepsilon_B)
        = \ext{\natepsilon_A}{A \cup B} \otimes \ext{\natepsilon_B}{A \cup B}= \natepsilon_{A \cup B} \otimes \natepsilon_{A \cup B} = \natepsilon_{A \cup B}$.
    If $a \in \gphi$,
    $a \otimes \natepsilon_\emptyset
        = (a \otimes \natepsilon_{\dd a}) \otimes \natepsilon_\emptyset
        = a \otimes (\natepsilon_{\dd a} \otimes \natepsilon_\emptyset)
        = a \otimes \natepsilon_{\dd a}
        = a
    $.
    Thus, $\natepsilon_\emptyset$ is a unit for $\gphi$.
\end{proof}

\begin{corollary}
    \label{cor:insertion_closure}
    Let $(\fphi, \otimes, \natepsilon)$ be an OVA, and let $B \subseteq A$ in $\cat{T}$, $a \in \fphi_A$, and $b \in \fphi_B$.
    Then,
    \begin{enumerate}
        \item   Restriction after extension is the identity map, i.e.,
              $
                  \proj{(\ext{b}{A})}{B} = b
              $.
        \item   Extension after restriction is \emph{extensive}, i.e.,
              $
                  a \leq_{\fphi_A} \ext{(\proj{a}{B})}{A}
              $.
    \end{enumerate}
\end{corollary}
\begin{proof}
    For the first claim, by the neutrality and combination axioms and naturality, we have
    $
        \proj{(\ext{b}{A})}{B}
        = \proj{(\natepsilon_A \otimes b)}{B}
        = \proj{\natepsilon_A}{A \cap B} \otimes b
        = \natepsilon_{B} \otimes b
        = b
    $.
    The second is always true of the composition of a right adjoint after its left adjoint.
\end{proof}

\begin{corollary}
    \label{cor:complete_lattice}
    If for each $A \in \cat{T}$, $\fphi_A$ is a complete lattice, then so is $\gphi$.
\end{corollary}
\begin{proof}
    See~\cite{DBLP:journals/tcs/TarleckiBG91}, where it is shown in more generality that completeness of the poset $\gphi$ follows from: (i) cocompleteness of the poset $\cat{T}$, (ii) completeness of each poset $\fphi_A$, and (iii) that the restriction maps $a \mapsto \proj{a}{B}$ have right adjoints.
\end{proof}

\begin{definition}[morphism of OVAs]
    \label{def:morphism_ova}
    Let $(\fphi, \otimes, \natepsilon)$ and $(\fphi', \otimes', \natepsilon')$ be OVAs.
    A \defn{lax morphism $f : \fphi \to \fphi'$} is a family of monotone maps $\set{f_A : \fphi_A \to \fphi'_A}_{A \in \cat{T}}$ so that the below hold for all $a \in \fphi_A$, $b \in \fphi_B$ and $C \subseteq A$,
    \begin{description}
        \item[Monotonicity.]
            \begin{equation}
                a \preceq b \implies f_A(a) \preceq f_B(b)
            \end{equation}
        \item[Lax naturality.]
            \begin{equation}
                \proj{f_A(a)}{C} \preceq f_C(\proj{a}{C})
            \end{equation}
        \item[Lax multiplicativity.]
            \begin{equation}
                f_A(a) \otimes' f_B(b) \preceq f_{A \cup B}(a \otimes b)
            \end{equation}
        \item[Lax unitality.]
            \begin{equation}
                \natepsilon'_A \preceq f_A (\natepsilon_A)
            \end{equation}
    \end{description}

    Reversing the inequality directions above defines a \defn{colax morphism}. A morphism that is both lax and colax is termed a \defn{strong morphism}.
\end{definition}

\subsection{Extension of local operators}

In the following, let $\fphi$ be a prealgebra such that for each inclusion $B \subseteq A$ in $\cat{T}$, the restriction map $a \mapsto \proj{a}{B} : \fphi_A \to \fphi_B$ has a right adjoint $b \mapsto \ext{b}{A} : \fphi_B \to \fphi_A$.

\begin{definition}[extension of a family of local operators]
    \label{def:ext_family}
    Assume a family of associative binary operators $
        \set{\odot_A : \fphi_A \times \fphi_A \to \fphi_A}_{A \in \cat{T}}
    $.
    We define the \defn{extension of $\set{\odot_A}_{A \in \cat{T}}$ to $\gphi$} to be the binary operator,
    \begin{equation}
        \begin{aligned}
            \odot     & : \gphi \times \gphi \to \gphi                                                      \\
            a \odot b & \defeq \ext{a}{\dd a \cup \dd b} \odot_{\dd a \cup \dd b} \ext{b}{\dd a \cup \dd b}
        \end{aligned}
    \end{equation}
\end{definition}

Note that a combine operator of an OVA $(\fphi, \otimes, \natepsilon)$ is the extension of the family $\set{\otimes_A}_{A \in \cat{T}}$, where $\otimes_A$ is the restriction of $\otimes$ to $\fphi_A \times \fphi_A$, because
\begin{equation}
    a \otimes b = (\natepsilon_{U} \otimes \natepsilon_{U}) \otimes (a \otimes b) = (\natepsilon_{U} \otimes a) \otimes (\natepsilon_{U} \otimes b) = \ext{a}{U} \otimes_{U} \ext{b}{U}
\end{equation}
where $U = \dd a\cup \dd b$.
Next is a key lemma establishing conditions for the reverse direction, i.e.~for when a family of local operators on $\fphi$ may give rise to a combine operator.

% % Omit this to save space
% \begin{remark}
%   The hypothesis of \cref{lem:helper} is not trivial or necessary, as we can see from the below example.
%   Let $X= \{x,y\}$ and $\cat{T} = \cat{2}^X$.
%   Let $\fphi = \fun{P} \circ \fomega$ be the valuation algebra of (ordinary) relations for the variables in $A \in \cat{T}$, where
%   \begin{equation}
%     \fomega = A \to \set{0,1}
%   \end{equation}

%   For each $A \in \cat{T}$, define $\odot_A$ to be the constant function on the zero tuple $\set{(0,\ldots,0)}$.
%   Clearly each $\odot_A$ is associative.
%   Then for $a, b \in \fphi_{\set{x}}$
%   \begin{equation}
%     \ext{a}{X} \odot_{X} \ext{b}{X}
%     =
%     {\set{(0,0)}}
%   \end{equation}
%   whereas
%   \begin{equation}
%     \ext{(a \odot_{\set{x}} b)}{X}
%     =\set{(0,0), (0,1)}
%   \end{equation}
%   Consequently, in reference to $\cref{lem:helper}$, extension-commutation is not a necessary condition for associativity or the combination axiom, as the above operator indeed satisfies the combination axiom and is monotone.
% \end{remark}

\begin{lemma}\ifarxiv[proof in~\cref{proof:helper}]\fi
    \label{lem:helper}
    Assume $\set{\odot_A : \fphi_A \times \fphi_A \to \fphi_A}_{A \in \cat{T}}$ is a family of local associative operators satisfying:
    \begin{description}
        \item[Local monotonicity.] For all $A \in \cat{T}$ and $a_1, a_1',a_2,a_2' \in \fphi_A$,
            \begin{equation}
                a_1 \leq_{\fphi_A} a_1' \text{ and } a_2 \leq_{\fphi_A} a_2' \implies
                a_1 \odot_A a_2 \leq_{\fphi_A} a_1' \odot_A a_2'
            \end{equation}
        \item[Extension-commutation.] For all $B \subseteq A$ in $\cat{T}$  and $b_1, b_2 \in \fphi_B$,
            \begin{equation}
                \ext{( b_1 \odot_B b_2)}{A}
                =
                \ext{b_1}{A} \odot_A \ext{b_2}{A}
            \end{equation}
    \end{description}
    Then $(\gphi, \odot)$ is an ordered semigroup, where $\odot$ is the extension of $\set{\odot_A}_{A \in \cat{T}}$.
\end{lemma}

The next lemma shows that to establish the \emph{weak exchange} axiom for a CVA (\cref{defn:cva} below), it suffices to show a local weak exchange law on each domain.

\begin{lemma}\ifarxiv[proof in~\cref{proof:local_weak_exchange}]\fi
    \label{lem:local_weak_exchange}
    Let $(\fphi, \parcomp, \neutralpar)$ and $(\fphi, \seq, \neutralseq)$ be OVAs whose combine operators $\parcomp$ and $\seq$ are respectively defined as extensions of $\set{\parcomp_A}_{A \in \cat{T}}$ and $\set{\seq_A}_{A \in \cat{T}}$.
    Assume that on each $A \in \cat{T}$, a weak exchange law holds: for all $a_1, a_2, a_3, a_4 \in \fphi_A$,
    $
        (a_1 \parcomp_A a_2) \seq_A (a_3 \parcomp_A a_4)
        \leq_{\fphi_A}
        (a_1 \seq_A a_3) \parcomp_A (a_2 \seq_A a_4)
    $.
    Then the weak exchange law holds on $\gphi$: for all $a,b,c,d \in \gphi$,
    $
        (a \parcomp b) \seq (c \parcomp d)
        \preceq
        (a \seq c) \parcomp (b \seq d)
    $.
\end{lemma}

%% file: content/cva.tex
\section{Concurrent valuation algebras}
\label{sec:cva}

We now introduce a \emph{concurrent valuation algebra (CVA)}, structured as two OVAs sharing the same underlying prealgebra, whose combine operators represent \emph{parallel} and \emph{sequential} products.
These operators are interlinked via a \emph{weak exchange law}, and their neutral elements are related by a pair of inequalities.

\begin{definition}[concurrent valuation algebra (CVA)]
    \label{defn:cva}
    A \emph{concurrent valuation algebra (CVA)} is a structure
    $(\fphi, \seq, \neutralseq, \parcomp, \neutralpar)$ satisfying the four axioms:
    \begin{description}
        \item[Sequential OVA.] $(\fphi, \seq, \neutralseq)$ is an OVA.\@
        \item[Parallel OVA.] $(\fphi, \parcomp, \neutralpar)$ is a commutative OVA.\@
        \item[Weak exchange.] For all $a,b,c,d \in \gphi$,
            \begin{equation}
                (a \parcomp b) \seq (c \parcomp d) \preceq (a \seq c) \parcomp (b \seq d)
            \end{equation}
        \item[Neutral laws.] For all $A \in \cat{T}$,
            \begin{equation}
                \neutralseq_A \preceq \neutralseq_A \parcomp \neutralseq_A
                \text{, and, }
                \neutralpar_A \seq \neutralpar_A \preceq \neutralpar_A
            \end{equation}
    \end{description}
\end{definition}

This definition is motivated by the relationship between sequential and parallel products. Sequential product, signifying a temporal juxtaposition, is generally noncommutative. In contrast, parallel product, signifying a spatial juxtaposition, is commutative. These two interlink by the weak exchange law. It states that the sequential composite of two parallel compositions, $a \parcomp b$ and $c \parcomp d$, results in fewer behaviours than the parallel composite of two sequential compositions, $a \seq c$ and $b \seq d$. Pictorially, this can be represented by a diagram (\cref{fig:weak_exchange}) where, on the left, $a$ and $b$ must finish together, causing $c$ and $d$ to start simultaneously. On the right, no such constraint is applied.
\begin{figure}[h]
    % https://q.uiver.app/#q=WzAsMjIsWzAsMF0sWzEsMF0sWzIsMF0sWzIsMV0sWzAsMV0sWzEsMV0sWzAsMl0sWzEsMl0sWzIsMl0sWzMsMSwiXFxwcmVjZXEiXSxbNCwyXSxbNywyXSxbNSwwXSxbNywwXSxbNiwyXSxbNSwxXSxbNiwxXSxbNCwwXSxbNCwxXSxbNywxXSxbNSwyXSxbNiwwXSxbMTIsMTUsIiIsMCx7InN0eWxlIjp7ImJvZHkiOnsibmFtZSI6ImRvdHRlZCJ9LCJoZWFkIjp7Im5hbWUiOiJub25lIn19fV0sWzE0LDE2LCIiLDAseyJzdHlsZSI6eyJib2R5Ijp7Im5hbWUiOiJkb3R0ZWQifSwiaGVhZCI6eyJuYW1lIjoibm9uZSJ9fX1dLFswLDUsImEiLDEseyJzdHlsZSI6eyJib2R5Ijp7Im5hbWUiOiJub25lIn0sImhlYWQiOnsibmFtZSI6Im5vbmUifX19XSxbMSwzLCJjIiwxLHsic3R5bGUiOnsiYm9keSI6eyJuYW1lIjoibm9uZSJ9LCJoZWFkIjp7Im5hbWUiOiJub25lIn19fV0sWzQsNywiYiIsMSx7InN0eWxlIjp7ImJvZHkiOnsibmFtZSI6Im5vbmUifSwiaGVhZCI6eyJuYW1lIjoibm9uZSJ9fX1dLFs1LDgsImQiLDEseyJzdHlsZSI6eyJib2R5Ijp7Im5hbWUiOiJub25lIn0sImhlYWQiOnsibmFtZSI6Im5vbmUifX19XSxbMTYsMTEsImQiLDEseyJzdHlsZSI6eyJib2R5Ijp7Im5hbWUiOiJub25lIn0sImhlYWQiOnsibmFtZSI6Im5vbmUifX19XSxbMSw3LCIiLDEseyJzdHlsZSI6eyJoZWFkIjp7Im5hbWUiOiJub25lIn19fV0sWzQsMywiIiwxLHsibGV2ZWwiOjIsInN0eWxlIjp7ImhlYWQiOnsibmFtZSI6Im5vbmUifX19XSxbMTcsMTUsImEiLDEseyJzdHlsZSI6eyJib2R5Ijp7Im5hbWUiOiJub25lIn0sImhlYWQiOnsibmFtZSI6Im5vbmUifX19XSxbMTgsMTksIiIsMCx7ImxldmVsIjoyLCJzdHlsZSI6eyJoZWFkIjp7Im5hbWUiOiJub25lIn19fV0sWzE1LDIwLCJiIiwxLHsic3R5bGUiOnsiYm9keSI6eyJuYW1lIjoibm9uZSJ9LCJoZWFkIjp7Im5hbWUiOiJub25lIn19fV0sWzIxLDE2LCJjIiwxLHsic3R5bGUiOnsiYm9keSI6eyJuYW1lIjoibm9uZSJ9LCJoZWFkIjp7Im5hbWUiOiJub25lIn19fV1d
    \[\begin{tikzcd}
            {} & {} & {} && {} & {} & {} & {} \\
            {} & {} & {} & \preceq & {} & {} & {} & {} \\
            {} & {} & {} && {} & {} & {} & {}
            \arrow[dotted, no head, from=1-6, to=2-6]
            \arrow[dotted, no head, from=3-7, to=2-7]
            \arrow["a"{description}, draw=none, from=1-1, to=2-2]
            \arrow["c"{description}, draw=none, from=1-2, to=2-3]
            \arrow["b"{description}, draw=none, from=2-1, to=3-2]
            \arrow["d"{description}, draw=none, from=2-2, to=3-3]
            \arrow["d"{description}, draw=none, from=2-7, to=3-8]
            \arrow[no head, from=1-2, to=3-2]
            \arrow[Rightarrow, no head, from=2-1, to=2-3]
            \arrow["a"{description}, draw=none, from=1-5, to=2-6]
            \arrow[Rightarrow, no head, from=2-5, to=2-8]
            \arrow["b"{description}, draw=none, from=2-6, to=3-6]
            \arrow["c"{description}, draw=none, from=1-7, to=2-7]
        \end{tikzcd}\]
    \caption{Graphical representation of the weak exchange law.}
    \label{fig:weak_exchange}
\end{figure}
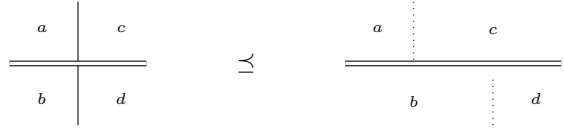

The neutral element of sequential composition, $\neutralseq$, acts as a null specification, thus $\neutralseq_A \parcomp \neutralseq_A$ must equal $\neutralseq_A$. Dually, the neutral element of parallel composition, $\neutralpar$, signifies an unconstrained specification, so $\neutralpar_A \seq \neutralpar_A$ equals $\neutralpar_A$. The proof of \cref{prop:neutrals} below shows that it's enough to assume one direction of these equalities; the other is derivable.

\begin{proposition}
    \label{prop:neutrals}
    In a CVA $(\fphi, \seq, \neutralseq, \parcomp, \neutralpar)$, for all $A \in \cat{T}$, we have $\neutralseq_A \preceq \neutralpar_A$, $\neutralseq_A \parcomp \neutralseq_A = \neutralseq_A$, and $\neutralpar_A \seq \neutralpar_A = \neutralpar_A$.
\end{proposition}
\begin{proof}
    By neutrality and weak exchange, for each $A \in \cat{T}$,
    $
        \neutralseq_A
        = \neutralseq_A \seq \neutralseq_A
        = (\neutralpar_A \parcomp \neutralseq_A) \seq (\neutralseq_A \parcomp \neutralpar_A)
        \preceq                (\neutralpar_A \seq \neutralseq_A) \parcomp (\neutralseq_A \seq \neutralpar_A)
        =                \neutralpar_A \parcomp \neutralpar_A
        = \neutralpar_A
    $.
    Given the neutral laws, for the remaining two properties, it suffices to show that $\neutralseq_A \parcomp \neutralseq_A \preceq \neutralseq_A$ and $\neutralpar_A \preceq \neutralpar_A \seq \neutralpar_A$.
    By monotonicity of combination, we have
    $
        \neutralseq_A \parcomp \neutralseq_A \preceq \neutralseq_A \parcomp \neutralpar_A = \neutralseq_A
    $.
    Similarly,
    $
        \neutralpar_A =
        \neutralseq_A \seq \neutralpar_A
        \preceq \neutralpar_A \seq \neutralpar_A
    $.
\end{proof}

% To-do: note that we only need to assume strong neutrality of one of the operations.
\begin{proposition}
    \label{prop:seq_le_par}
    In a CVA $(\fphi, \seq, \natepsilon, \parcomp, \natepsilon)$ in which the neutral elements of parallel and sequential product coincide, for all $a, b \in \gphi$, we have $a \seq b \preceq a \parcomp b$.
\end{proposition}
\begin{proof}
    Let $a \in \fphi_A$ and $b \in \fphi_B$.
    We have,
    $
        a \seq b
        = \ext{a}{A \cup B} \seq \ext{b}{A \cup B}
        = (a \parcomp \natepsilon_{A \cup B}) \seq (\natepsilon_{A \cup B} \parcomp b)
        \preceq (a \seq \natepsilon_{A \cup B}) \parcomp (\natepsilon_{A \cup B} \seq b)
        = \ext{a}{A \cup B} \parcomp \ext{b}{A \cup B}
        = a \parcomp b
    $.
\end{proof}

\begin{definition}[morphism of CVAs]
    \label{def:morphism_cva}
    Let $(\fphi, \seq, \neutralseq, \parcomp, \neutralpar)$ and $(\fphi', \seq', \neutralseq', \parcomp', \neutralpar')$ be CVAs.
    A \defn{lax/colax/strong morphism $f: \fphi \to \fphi'$} is a function $f : \gphi \to \gphi'$ that is both a lax/colax/strong morphism of OVAs $(\fphi, \seq, \neutralseq) \to (\fphi', \seq', \neutralseq)$ and a lax/colax/strong morphism of OVAs $(\fphi, \parcomp, \neutralpar) \to (\fphi', \parcomp', \neutralpar)$.\footnote{In duoidal categories, morphisms may also be lax with respect to $\seq$ and colax with respect to $\parcomp$, but not the reverse~\cite{MR2724388}.}
\end{definition}

\subsection{Reasoning in a CVA}

\subsubsection{Refinement.}
\label{sec:refinement}

In a CVA $\fphi$, the ordering between elements $a$ and $b$ in $\gphi$ is defined as $a \preceq b$ if and only if $\dd b \subseteq \dd a$ and $\proj{a}{\dd b} \leq_{\fphi_{\dd b}} b$. Viewing these elements as system specifications, this ordering is interpreted as refinement: $a \preceq b$ means that all behaviour of $a$ within domain $\dd b$ also exists in $b$, making $a$ on $\dd b$ more deterministic than $b$. However, the domain $\dd a$ of $a$ may exceed $\dd b$, as a refined specification may introduce constraints outside the initial domain.

\subsubsection{Hoare logic and rely-guarantee reasoning.}

\emph{Hoare triples} and \emph{Jones quintuples} facilitate formal reasoning about program behaviour, leveraging the well-established methodologies of Hoare logic and rely-guarantee reasoning.
These constructs may be realised in a CVA by adapting their definitions as framed within Concurrent Kleene Algebras~\cite{DBLP:journals/jlp/HoareMSW11}.\@

Let $(\fphi, \seq, \neutralseq, \parcomp, \neutralpar)$ be a CVA, and $p,a,q \in \gphi$.
We define the \defn{Hoare triple} of $a$ with \emph{precondition} $p$ and \emph{postcondition} $q$ as
\begin{equation}
    \hoare{p}{a}{q} \defeq p \seq a \preceq q
\end{equation}
From this definition, we may derive inference rules\footnote{Other basic rules are verified in the computer formalisation (see \cref{fn:isabelle}).} of Hoare logic, such as:

\begin{proposition}[concurrency rule]
    Let $p,p',a,a',q,q' \in \gphi$.
    Then
    \begin{equation}
        \hoare{p}{a}{q} \text{ and } \hoare{p'}{a'}{q'} \implies \hoare{(p \parcomp p')}{a \parcomp a'}{(q \parcomp q')}
    \end{equation}
\end{proposition}
\begin{proof}
    Assume $p \seq a \preceq q$ and $p' \seq a' \preceq q'$.
    By weak exchange and monotonicity,
    $
        (p \parcomp p') \seq (a \parcomp a')
        \preceq (p \seq a) \parcomp (p' \seq a')
        \preceq q \parcomp q'
    $.
    Thus, $\hoare{(p \parcomp p')}{a \parcomp a'}{(q \parcomp q')}$.
\end{proof}
A \defn{Jones quintuple} with \emph{rely} $r$ and \emph{guarantee} $g$ can then be defined as\footnote{The guarantee requirement is stronger than required by Jones, where the guarantee only must hold while the rely does.}
\begin{equation}
    \jones{p}{r}{a}{g}{q}
    \defeq \hoare{p}{r \parcomp a}{q} \text{ and } a \preceq g
\end{equation}
To employ the standard inference rules of rely-guarantee reasoning, constraints must be placed on the rely variable $r$ and the guarantee variable $g$.
Though this definition serves as a gateway to rely-guarantee reasoning in the context of a CVA, exploration of this aspect is beyond the present study's purview.

%% file: content/tuple_systems.tex
\section{Tuple systems}
\label{sec:tuple}

In~\cref{sec:action_trace_model,sec:state_trace_model,sec:relative_state_model}, each CVA examined is based on an underlying OVA of a specific form---they are \emph{OVAs of $\fun{T}$-relations} associated to certain \emph{tuple systems} $\fun{T}$.
Tuple systems are presheaves that abstract the characteristic projecting and lifting properties of ordinary tuples.
For more on tuple systems and the valuation algebras they induce, please see~\cite[Section 6.3, p.\@ 169]{DBLP:books/daglib/0008195} and~\cite[Section 7.3.2, p.\@ 286]{pouly2012generic}.
A $\fun{T}$-relation is a subset of these generalised \emph{tuples} sharing a common domain.
In the trace models to follow, actions, states, traces and valuations themselves are encoded as tuples within tuple systems.
The structure of the tuple system $\fun{T}$ governs how tuples on a larger domain project to a smaller one through the presheaf's restriction maps, as well as how tuples on a smaller domain lift to a larger one via the presheaf's flasque and binary gluing properties.

\begin{definition}[tuple system]
    A \defn{tuple system} is a presheaf $\fun{T} : \opcat{T} \to \cat{Set}$ satisfying the below axioms:
    \begin{description}
        \item[Flasque.] For all $B \subseteq A$ in $\cat{T}$, the restriction map $\fun{T}_A \to \fun{T}_B$ is surjective.
        \item[Binary gluing.] For all $a \in \fun{T}_A $ and $b \in \fun{T}_B$, if $\projtup{a}{A \cap B} = \projtup{b}{A \cap B}$, then there exists $c \in \fun{T}_{A \cup B}$ so that $\projtup{c}{A}= a$ and $\projtup{c}{B} = b$.
    \end{description}
    Elements of $\fun{T}_A$ are called \defn{tuples (on $A$)} or \defn{$A$-tuples}.
\end{definition}

\begin{theorem}[OVAs of $\fun{T}$-relations]
    \label{thm:rel_ova}
    Let $\fun{T} : \opcat{T} \to \cat{Set}$ be a \defn{tuple system}.
    Define the prealgebra
    $
        \fpsi \defeq \fun{P} \circ \fun{T} : \opcat{T} \to \cat{Pos}
    $.
    Then $\fpsi$, equipped with the \defn{relational join} as the combine operator, defined
    \begin{equation}
        \begin{aligned}
            \label{eq:rel_join}
            \natjoin & : \gpsi \times \gpsi \to \gpsi \\
            a \natjoin b
                     & \defeq
            \setc{ t \in \fun{T}_{\dd a \cup \dd b} }{
                \projtup{t}{\dd a} \in a
                ,
                \projtup{t}{\dd b} \in b
            }
        \end{aligned}
    \end{equation}
    is a strongly neutral commutative OVA, that we call the \defn{OVA of $\fun{T}$-relations}.
    Its local orderings $\leq_{\fun{T}_A}$ are given by subset inclusion $\subseteq$, and it has as neutral element $\nattop = A \mapsto \fun{T}_A$ for each $A \in \cat{T}$.
    Moreover, $\fun{U} \circ \fpsi$ is itself a tuple system, where $\fun{U} : \cat{Pos} \to \cat{Set}$ is the forgetful functor that sends a poset to its underlying set, and a monotone map to its underlying function.\footnote{This last point follows from the \emph{idempotence} property of $\natjoin$~\cite[Example 7.7, p.\@ 287]{pouly2012generic}.}
\end{theorem}
\begin{proof}
    Monotonicity is easily verified.
    The other details are found in~\cite[p.\@ 170]{DBLP:books/daglib/0008195}.
\end{proof}

It is worth noting the close resemblance of \cref{eq:rel_join} with the trace semantics of the CSP parallel operator~\cite[Section 2.3.3, p. 53]{DBLP:books/ph/Hoare85}.

\begin{proposition}
    \label{prop:tuple_ext}
    Extension is given by the preimage to restriction; i.e.~for $ a \in \fphi_A$, and $B \in \cat{T}$ with $B \subseteq A$, we have
    $
        \ext{ b }{A}
        =
        \setc{t \in \fun{T}_A}{
            \projtup{t}{B} \in b
        }
    $.
\end{proposition}
\begin{proof}
    It is a standard proof that direct image is left-adjoint to preimage.
\end{proof}

\begin{proposition}\ifarxiv[proof in~\cref{proof:meet}]\fi
    \label{prop:rel_meet}
    The relational join of an OVA of relations is the extension of intersection (from \cref{def:ext_family}): for $a,b \in \gphi$,
    $
        a \natjoin b =
        \ext{ a }{\dd a \cup \dd b} \cap \ext{ b}{\dd a \cup \dd b}
    $.
    Moreover, $\gphi$ is a complete lattice, and relational join is its meet.
\end{proposition}

\begin{lemma}\ifarxiv[proof in~\cref{proof:tuple_lists}]\fi
    \label{lem:tuple_lists}
    Let $\fomega : \opcat{T} \to \cat{Set}$ be a tuple system, and let $\lists : \cat{Set} \to \cat{Set}$ be the functor that sends a set $X$ to the set of finite lists in $X$, i.e.
    $
        \lists \defeq X \mapsto \coprod_{n \in \Nn} X^n
    $,
    and let $\listsi$ be the functor that sends $X$ to the set of nonempty finite lists in $X$, i.e.
    $
        \listsi \defeq X \mapsto \coprod_{n \in {\NnOne}} X^n
    $.
    Then both $\lists \circ \fomega$ and $\listsi \circ \fomega$ are tuple systems.
\end{lemma}

\begin{notation}
    Square brackets are used to display the components of a tuple $t \in {(\fomega_A)}^n$, i.e.~we write $t = [t_1, \ldots, t_n]$.
    Such tuples are referred to as \defn{traces}.
\end{notation}

%% file: content/action_trace.tex
\section{Action trace model}
\label{sec:action_trace_model}

Let $\aomega : \opcat{T} \to \cat{Set}$ be a tuple system whose values $\aomega_A$ represent possible \emph{actions} of a system in the variables $A$.
Some concrete examples: for a semiring $\Ss$ of \emph{values}, $\aomega_A$ is the set of matrices $A \times A \to \Ss$ (\emph{linear actions}); the set of pairs $\Ss^A \times \Ss^A$ (\emph{events}); the set of relations $\fun{P}(\Ss^A \times \Ss^A)$ (\emph{events with external choice}).
Let $\atup \defeq \lists \circ \aomega$, so that for each $A \in \cat{T}$, $\atup_A$ is the set of (possibly empty) traces of elements of $\aomega_A$.
By~\cref{lem:tuple_lists}, $\atup$ is a tuple system.
Let
\begin{equation}
    \begin{aligned}
        \faction & : \opcat{T} \to \cat{Pos}                                          \\
        \faction & \defeq \fun{P} \circ \atup = A \mapsto \fun{P} (\lists(\aomega_A))
    \end{aligned}
\end{equation}
be the OVA of $\atup$-relations.
We now develop a CVA structure on $\faction$ that we call the \defn{action trace model}.\@

For each $A \in \cat{T}$, define
\begin{equation}
    \emp_A \defeq \set{ \unit_A }
\end{equation}
where $\unit_A  \in \atup_A$ is the unique length-0 trace with domain $A$.
As restriction of a trace preserves length, this defines a global element $\emp : \fun{1} \nattra \faction$.

\subsection{Interleaving product}

For all $p,q \in \Nn$, let $\Sigma_{p,q}$ be the set \emph{$(p,q)$-shuffles}, i.e.~bijections $\set{1,\ldots,p+q} \to \set{1,\ldots,p+q}$ (or permutations) such that $\sigma(1) < \cdots < \sigma(p)$ and $\sigma(p+1) < \cdots < \sigma(p+q)$.
For each $A \in \cat{T}$, define an operator on traces,
\begin{equation}
    \begin{aligned}
        \shuftrace{A} & : \atup_A \times \atup_A \to \faction_A                                          \\
        [t_1,\ldots,t_p] \shuftrace{A} [t_{p+1},\ldots,t_{p+q}]
                      & \defeq \{ [t_{\sigma(1)},\ldots,t_{\sigma(p+q)}] \mid \sigma \in \Sigma_{p,q} \}
    \end{aligned}
\end{equation}
Then lift each $\shuftrace{A}$ to a local operator on valuations,
\begin{equation}
    \begin{aligned}
        \shufdom{A}      & : \faction_A \times \faction_A \to \faction_A                              \\
        a \shufdom{A} a' & \defeq \bigcup \setc{t_a \shuftrace{A} t_{a'} }{t_a \in a , t_{a'} \in a'}
    \end{aligned}
\end{equation}
It is well-known that $\shufdom{A}$ is commutative, associative, and has unit $\emp_A$.
We then define the \defn{interleaving product} as the extension $\shuf$ of $\set{\shufdom{A}}_{A \in \cat{T}}$ to $\gaction$:
\begin{equation}
    \begin{aligned}
        \shuf     & : \gaction \times \gaction \to \gaction                                                \\
        a \shuf b & \defeq \ext{a}{\dd a \cup \dd b } \shufdom{\dd a \cup \dd b} \ext{b}{\dd a \cup \dd b}
    \end{aligned}
\end{equation}
Note that $\shuf$ is clearly commutative, and has as neutral element $\emp$.

\begin{lemma}\ifarxiv[proof in~\cref{proof:shuf_commute_proj}]\fi
    \label{lem:shuf_commute_proj}
    For all $t,s \in \atup_A$ and $B \subseteq A$, we have
    $
        \proj{(t \shuftrace{A} s)}{B} =
        \projtup{t}{B} \shuftrace{B} \projtup{s}{B}
    $.
\end{lemma}

\begin{lemma}
    \label{lem:shuf_ordered_semigroup}
    The structure $(\gaction, \shuf)$ is an ordered semigroup.
\end{lemma}
\begin{proof}
    By \cref{lem:helper},
    it suffices to show that the local monotonicity and extension-commutation properties hold.
    The former follows directly from the definition of $\shufdom{A}$.
    For extension-commutation, let $B \subseteq A$, let $b, b' \in \faction_B$, and let $t \in \ext{b}{A} \shufdom{A} \ext{b'}{A}$.
    By definition of $\shufdom{A}$, there exists $r \in \ext{b}{A}$, $s \in \ext{b'}{A}$ so that $t \in r \shuftrace{A} s$.
    By \cref{lem:shuf_commute_proj}, $\projtup{t}{B} \in
        \proj{(r \shuftrace{A} s)}{B} =
        \projtup{r}{B} \shuftrace{B} \projtup{s}{B} \subseteq b \shufdom{B} b'$.
    Thus, $t \in \ext{( b \shufdom{B} b')}{A}$.
    Conversely, let $t' \in \ext{( b \shufdom{B} b')}{A}$.
    Now there is $r \in b$, $s \in b'$ so that $\projtup{t'}{B} \in r \shuftrace{B} s$.
    We may write $r = [t_1, \ldots, t_p]$, $s = [t_{p+1},\ldots, t_{p+q}]$, and $\projtup{t'}{B} = [t_{\sigma(1)}, \ldots, t_{\sigma(p+q)}]$ for a $(p,q)$-shuffle $\sigma \in \Sigma_{p,q}$.
    For each $1 \leq i \leq p + q$, we then have a lifting $t_{\sigma(i)}'$ of $t_{\sigma(i)}$ so that $t = [t_{\sigma(1)}', \ldots, t_{\sigma(p+q)}']$.
    Then $r' = [t'_1,\ldots,t'_p]$ is a lifting of $r$, $s' = [t'_{p+1}, \ldots, t'_{p+q}]$ is a lifting of $s$, and $t' \in r' \shuftrace{A} s' \in \ext{b}{A} \shufdom{A} \ext{b'}{A}$ is exhibited as a trace associated to the same $(p,q)$-shuffle $\sigma$.
    The result follows.
\end{proof}

\begin{lemma}
    \label{lem:shuf_comb}
    The interleaving product $\shuf$ satisfies the combination axiom.
\end{lemma}
\begin{proof}
    Let $A, B \in \cat{T}$, $a \in \faction_A$ and $b \in \faction_B$.
    Note that one direction of the combination law follows from monotonicity.
    It then suffices to show $a \shuf \proj{ b }{A \cap B } \subseteq \proj{ ( a \shuf b ) }{A}$ and $\proj{ a }{A \cap B } \shuf b \subseteq \proj{ ( a \shuf b ) }{B}$.
    Let $t \in a \shuf \proj{ b }{A \cap B } = a \shufdom{A} \ext{(\proj{b}{A \cap B})}{A}$.
    By definition of $\shufdom{A}$, there exists $t_a \in a$ and $t_b \in \ext{(\proj{b}{A \cap B})}{A}$ so that $t \in t_a \shuftrace{A} t_b$.
    Let $t_a' \in \ext{a}{A \cup B}$ be a lifting of $t_a$.
    % By~\cref{lem:extension_int_union}, $\ext{(\proj{b}{A \cap B})}{A} = \proj{(\ext{b}{A \cup B})}{A}$, so let $t_b' \in \ext{b}{A \cup B}$ be a lifting of $t_b$.
    Let
    $
        s \defeq
        \projtup{t_b}{A \cap B}
        \in
        \proj{(\ext{(\proj{b}{A \cap B})}{A})}{A \cap B}
        = \proj{b}{A \cap B}
    $,
    where the equality follows by~\cref{cor:insertion_closure}.
    There then exists $s' \in b$ so that $\projtup{s'}{A \cap B} = s$.
    By binary gluing, there exists a common lifting $t_b' \in \ext{b}{A \cup B}$ of $t_b$ and $s'$.
    As in the proof of~\cref{lem:shuf_ordered_semigroup}, it is easily shown that there is $t' \in t_a' \shuftrace{A \cup B} t_b' \in a \shuf b$ (associated to the same $(p,q)$-shuffle as $t$) so that $t =\projtup{t'}{A} \in \proj{ ( a \shuf b ) }{A}$.
    Similarly, $\proj{ a }{A \cap B } \shuf b \subseteq \proj{ ( a \shuf b ) }{B}$.
    The result follows.
\end{proof}

As strong neutrality easily holds, we have the following.

\begin{proposition}
    \label{prop:shuf_ova}
    The structure $(\faction, \shuf, \emp)$ is a strongly neutral commutative OVA.\@
\end{proposition}

\subsection{Concatenating product}

For each $A \in \cat{T}$, define the associative binary operator on traces,
\begin{equation}
    \begin{aligned}
        \conctrace{A} & : \atup_A \times \atup_A \to \atup_A \\
        [t_1,\ldots,t_n] \conctrace{A} [s_1,\ldots,s_m]
                      & \defeq
        [ t_1, \ldots, t_n, s_1, \ldots, s_m]
    \end{aligned}
\end{equation}
Then lift each $\conctrace{A}$ to a local operator on valuations,
\begin{equation}
    \begin{aligned}
        \concdom{A} & : \faction_A \times \faction_A \to \faction_A \\
        a \concdom{A} a'
                    & \defeq
        \setc{
            t_a \conctrace{A} t_{a'}
        }{
            t_a \in a, t_{a'}\in a'
        }
    \end{aligned}
\end{equation}
We call the extension $\conc$ to $\gaction$ of $\set{\concdom{A}}_{A \in \cat{T}}$ the \defn{concatenating product}.

\begin{proposition}\ifarxiv[proof in~\cref{proof:conc_ova}]\fi
    \label{prop:conc_ova}
    The structure $(\faction, \conc, \emp)$ is a strongly neutral OVA.\@
\end{proposition}

\begin{proposition}
    \label{prop:action_cva}
    The structure $\tuple{\faction, \conc, \emp, \shuf, \emp }$ is a CVA.\@
\end{proposition}
\begin{proof}
    Both $\conc$ and $\shuf$ define OVA structures on $\faction$ (\cref{prop:shuf_ova,prop:conc_ova}), and the neutral laws $\emp_A \subseteq \emp_A \shuf \emp_A$ and $\emp_A \conc \emp_A \subseteq \emp_A$ hold trivially.
    To show the weak exchange law, by~\cref{lem:local_weak_exchange}, it suffices to show a local exchange law holds on each $A \in \cat{T}$.
    Let $a_1,a_2,a_3,a_4 \in \faction_A$ and $t \in ( a_1 \shufdom{A} a_2 ) \concdom{A} ( a_3 \shufdom{A} a_4 )$.
    By definition of $\conc$, there is $r \in a_1 \shufdom{A} a_2$ and $s \in a_3 \shufdom{A} a_4$ so that $t = r \conctrace{A} s$.
    It is clear every action of $t$ coming from $a_1$ precedes every action of $t$ coming from $a_3$, and similarly every action of $t$ coming from $a_2$ precedes every action of $t$ coming from $a_4$.
    It follows that $t$ is in $(a_1 \concdom{A} a_3) \shufdom{A} (a_2 \concdom{A} a_4)$.
    The result follows.
\end{proof}

\begin{proposition}
    For all $a, b \in \gphi$, we have $a \conc b \preceq a \shuf b$.
\end{proposition}
\begin{proof}
    As the units for $\conc$ and $\shuf$ coincide, this follows from~\cref{prop:seq_le_par}.
\end{proof}

%% file: content/state_trace.tex
\section{State trace model}
\label{sec:state_trace_model}
Here we define a CVA whose valuations consist of traces of \emph{states} of an abstract system that progress in lockstep to an implied global clock.
For each domain $A \in \cat{T}$,
denote the hom-functor
$
    \somega \defeq A \mapsto (A \to \Ss)
$,
i.e., $\somega_A$ is the set of (ordinary) $A$-tuples in some nonempty set $\Ss$ of \emph{values}, and the action of $\somega$ on inclusions in $\cat{T}$ is by precomposition.
Notably, $\somega_\emptyset$ has a unique value $\emptystate$, the \emph{empty state}.
By~\cref{lem:tuple_lists}, $\stup \defeq \listsi \circ \somega$ is a tuple system.
For traces $t \defeq [t_1,\ldots,t_n] \in \stup_A$, a component $t_i$ is the system's state at time $i$.
Let
\begin{equation}
    \begin{aligned}
        \fstate & : \opcat{T} \to \cat{Pos}                                               \\
        \fstate & \defeq \fun{P} \circ \stup  =  A \mapsto \fun{P} (\listsi (A \to \Ss) )
    \end{aligned}
\end{equation}
be the OVA of $\stup$-relations.
The relational join $\natjoin$ on $\fstate$ behaves as \emph{synchronisation}, and we take this as the parallel product for a CVA structure on $\fstate$ that we call the \defn{state trace model}.

Let $\lambda : \stup_A \to \NnOne$ denote the length function, and define
\begin{equation}
    \nil \defeq A \mapsto \setc{t \in \stup_A}{\lambda(t) = 1}
\end{equation}
As restriction preserves lengths of traces, this defines a global element $\nil : \fun{1} \nattra \fstate$.

\subsection{Gluing product}
\label{sec:glue_ova}

Let $A \in \cat{T}$.
For a trace $t \in \stup_A$, let $t^-, t^+$ respectively denote the first and last components of $t$.
Define an associative binary operator $\gluetrace{A}$ on each $\stup_A$ by
\begin{equation}
    \begin{aligned}
        \gluetrace{A} & : \stup_A \times \stup_A \to \stup_A \\
        t \gluetrace{A} s
                      & \defeq
        ( t_1, \ldots, t_{\lambda(t)-1}, s_1, \ldots, s_{\lambda(s)})
    \end{aligned}
\end{equation}
This then lifts to an associative binary operator on valuations,
\begin{equation}
    \begin{aligned}
        \gluedom{A} & : \fstate_A \times \fstate_A \to \fstate_A \\
        a \gluedom{A} a'
                    & \defeq
        \setc{
            t_a \gluetrace{A} t_{a'}
        }{
            t_a \in a , t_{a'} \in a' ,
            t_a^+ = t_{a'}^-
        }
    \end{aligned}
\end{equation}
We call the extension $\glue$ to $\gstate$ of the family $\set{\gluedom{A}}_{A \in \cat{T}}$ the \defn{gluing product}.

\begin{proposition}\ifarxiv[proof in~\cref{proof:glue_ova}]\fi
    \label{prop:glue_ova}
    The structure $(\fstate, \glue, \nil)$ is a strongly neutral OVA.\@
\end{proposition}

\begin{proposition}
    \label{prop:state_cva}
    The structure $\tuple{\fstate, \glue,\nil, \natjoin, \ski }$ is a CVA.\@
\end{proposition}
\begin{proof}
    The neutral equalities are clear, and both $\natjoin$ and $\glue$ define OVAs on $\fstate$ by~\cref{thm:rel_ova,prop:glue_ova}.
    By \cref{lem:local_weak_exchange}, it suffices to show an exchange law holds on each $A \in \cat{T}$.
    Noting that $\natjoin$ is the extension of intersection by~\cref{prop:rel_meet}, let $a_1,a_2,a_3,a_4 \in \fstate_A$ and let $t \in ( a_1 \cap a_2 ) \gluedom{A} ( a_3 \cap a_4 )$.
    By local monotonicity of $\gluedom{A}$, both $t \in a_1 \gluedom{A} a_3$ and $ t \in a_2 \gluedom{A} a_4$.
    Thus, $t \in (a_1 \gluedom{A} a_3) \cap (a_2 \gluedom{A} a_4)$, and the result follows.
\end{proof}

\subsection{Strong morphisms between $\faction$ and $\fstate$}

There are no interesting strong morphisms between the action trace model $\faction$ and the state trace model $\fstate$.
As the neutral elements for parallel and sequential coincide in $\faction$ but not in $\fstate$, there are no strong morphisms $\faction \to \fstate$.
On the other hand, a strong morphism $f : \fstate \to \faction$ must map $\ski_A$ to $\emp_A$, and by monotonicity this implies that $f_A(a) \subseteq \emp_A$ for all $a \in \fstate_A$.
Whether there are interesting (co)lax morphisms between $\faction$ and $\fstate$ is an open question.

%% file: content/relative_trace.tex
\section{Relative state trace model}
\label{sec:relative_state_model}

We introduce a variant, $\fstaterel$, of the state trace model from \cref{sec:state_trace_model}, that we refer to as the \defn{relative state trace model}. In this model, traces are \emph{stuttering-reduced}, meaning they do not contain duplicate adjacent components. Consequently, only the relative order of the indices in the trace components is significant, indicating independence from a global clock. This may lead to intriguing phenomena like \emph{sequential inconsistency}~\cite{Evangelou23}.

An essentially equivalent construction of the underlying relational OVA was already presented in~\cite{Evangelou23} using simplicial sets. Here, we offer a more concise and direct method using \emph{free semigroups with idempotent generators}, previously applied to concurrency theory and quantum computation~\cite{DBLP:journals/nc/BertoniMP10}.

Let $S$ be a set.
Construct a semigroup $\fun{I}(S)$ as the free semigroup on $S$ modulo the relation $x^2 = x$ for all $x \in S$.
This is known as the \defn{free semigroup on $S$ with idempotent generators}.
For example, if $S \defeq \set{0,1}$, then $\fun{I}(S) = \set{0, 1, 01, 10, 010, 101, 0101, \ldots}$, and the semigroup product is concatenation modulo this congruence; e.g., $010 \cdot 01 = 0101$.
Given a function $f:S \to S'$, there is a semigroup homomorphism $\fun{I}(f) :\fun{I}(S) \to \fun{I}(S')$, defined by $\fun{I}(f)(x_1 \cdots x_n) \defeq f(x_1) \cdots f(x_n)$, and moreover, this construction is functorial.
Let $\fun{U} : \cat{Semi} \to \cat{Set}$ be the forgetful functor from the category of semigroups to the category of sets, that sends a semigroup to its underlying set, and a semigroup homomorphism to its underlying function.
As in~\cref{sec:state_trace_model}, let $\somega$ be the contravariant hom-functor
$
    \somega = A \mapsto (A \to \Ss)
$
where $\Ss$ is a fixed set of values.
We then define
$
    \rtup \defeq \fun{U} \circ \fun{I} \circ \somega : \opcat{T} \to \cat{Set}
$.

\begin{proposition}
    The presheaf $\rtup$ is a tuple system.
\end{proposition}
\begin{proof}[sketch]
    This is essentially equivalent to~\cite[Theorem 2]{Evangelou23}.
    There, empty traces were included in the tuple system by use of the \emph{augmented} simplicial nerve functor. If the ordinary nerve were used, the same proof goes through, and we would exclude empty traces (problematic here in defining gluing product), yielding a tuple system isomorphic to the one described here with semigroups.
\end{proof}

Now let $\fstaterel \defeq \fun{P} \circ \rtup$ be the OVA of $\rtup$-relations, and denote the relational join $\rnatjoin$ and its neutral element $\rski = A \mapsto \rtup_A$.
Note that while $\ski_\emptyset = \stup_\emptyset$ has infinitely many elements $[\emptystate], [\emptystate,\emptystate], \ldots$, the neutral component $\rski_\emptyset = \rtup_\emptyset$ has only one, namely $[\emptystate]$.
We define a local operator on valuations,
\begin{equation}
    \begin{aligned}
        \rgluedom{A} & : \fstaterel_A \times \fstaterel_A \to \fstaterel_A \\
        a \rgluedom{A} a'
                     & \defeq
        \setc{
            t_a \rgluetrace{A} t_{a'}
        }{
            t_a \in a , t_{a'} \in a' ,
            t_a^+ = t_{a'}^-
        }
    \end{aligned}
\end{equation}
where $\rgluetrace{A}$ is the product\footnote{\label{fn:forget}To avoid excessive notation, we apply the semigroup products $\cdot_A$ directly to traces, although their semigroup structure was forgotten by $\fun{U}$.} of the semigroup $\fun{I}(\somega_A)$, and $\lambda$ and $t \mapsto t^+, t^-$ are defined as in~\cref{sec:state_trace_model}.
We call the extension $\rglue$ of $\set{\rgluedom{A}}_{A \in \cat{T}}$ to $\gstaterel$ the \defn{relative gluing product}.
Let $\rnil \defeq A \mapsto \setc{t \in \rtup_A}{\lambda(t) = 1}$.
We then have,

\begin{proposition}\ifarxiv[proof in~\cref{proof:relative_glue_ova}]\fi
    \label{prop:relative_glue_ova}
    The structure $(\fstaterel, \rglue, \rnil)$ is an OVA.\@
\end{proposition}

Unlike the models $\faction$ and $\fstate$ of~\cref{sec:action_trace_model,sec:state_trace_model}, we have the following.
\begin{proposition}
    \label{prop:relative_glue_ova_not_strongly_neutral}
    The OVA $(\fstaterel, \rglue, \rnil)$ is not strongly neutral.
\end{proposition}
\begin{proof}
    We have$\rski_\emptyset = \rnil_\emptyset$ and yet $\rski \neq \rnil$.
    The result follows.
\end{proof}

\begin{proposition}
    \label{prop:relative_cva}
    The structure $(\fstaterel,\rglue, \rnil, \rnatjoin, \rski )$ is a CVA.\@
\end{proposition}
\begin{proof}
    \label{proof:relative_cva}
    The neutral laws are immediate, and we are only obliged to show the local weak exchange laws hold by~\cref{prop:rel_meet,lem:local_weak_exchange}.
    Locally $\rnatjoin_A = \natjoin_A = \cap$, and also $\gluedom{A}$ and $\rgluedom{A}$ have the same effect on traces, i.e.~the gluing of two stuttering-reduced traces is already stuttering-reduced, so the proof of~\cref{prop:state_cva} goes through unchanged.
\end{proof}

\subsection{Colax morphism from $\fstate$ to $\fstaterel$}
\label{sec:relative_morphism}

Define the free semigroup functor $\fun{F} : \cat{Set} \to \cat{Semi}$ mapping set $S \in \cat{Set}$ to finite lists of its elements, using concatenation as the semigroup product. Please note there is an evident isomorphism $\listsi \cong \fun{U} \circ \fun{F}$ that we will apply implicitly. The universal property of the free semigroup leads to a surjective map $q_S : \fun{F}(S) \epi \fun{I}(S)$ for each set $S$, which acts to eliminate duplicated adjacent elements in a list. This process defines a natural transformation $q : \fun{F} \nattra \fun{I}$, allowing us to obtain another natural transformation by whiskering\footnote{See~\cite[Remark 1.7.6., p.46]{riehl2017category}.} on both sides of $q$.

% https://q.uiver.app/#q=WzAsNSxbMiwwLCJcXGNhdHtTZW1pfSJdLFszLDAsIlxcY2F0e1NldH0iXSxbMSwwLCJcXGNhdHtTZXR9Il0sWzAsMCwiXFxvcGNhdHtUfSJdLFs0LDAsIlxcY2F0e1Bvc30iXSxbMCwxLCJcXGZ1bntVfSJdLFsyLDAsIlxcZnVue0l9IiwyLHsiY3VydmUiOjJ9XSxbMywyLCJcXHNvbWVnYSJdLFsyLDAsIlxcZnVue0Z9IiwwLHsiY3VydmUiOi0yfV0sWzEsNCwiXFxmdW57UH0iXSxbMiwxLCJcXGxpc3RzaSIsMCx7ImN1cnZlIjotNX1dLFs4LDYsInEiLDAseyJzaG9ydGVuIjp7InNvdXJjZSI6MjAsInRhcmdldCI6MjB9fV0sWzEwLDAsIlxcY29uZyIsMyx7InNob3J0ZW4iOnsic291cmNlIjoyMH0sInN0eWxlIjp7ImJvZHkiOnsibmFtZSI6Im5vbmUifSwiaGVhZCI6eyJuYW1lIjoibm9uZSJ9fX1dXQ==
\begin{equation}
    \begin{tikzcd}
        {\opcat{T}} & {\cat{Set}} & {\cat{Semi}} & {\cat{Set}} & {\cat{Pos}}
        \arrow["{\fun{U}}", from=1-3, to=1-4]
        \arrow[""{name=0, anchor=center, inner sep=0}, "{\fun{I}}"', curve={height=12pt}, from=1-2, to=1-3]
        \arrow["\somega", from=1-1, to=1-2]
        \arrow[""{name=1, anchor=center, inner sep=0}, "{\fun{F}}", curve={height=-12pt}, from=1-2, to=1-3]
        \arrow["{\fun{P}}", from=1-4, to=1-5]
        \arrow[""{name=2, anchor=center, inner sep=0}, "\listsi", curve={height=-40pt}, from=1-2, to=1-4]
        \arrow["q", shorten <=3pt, shorten >=3pt, Rightarrow, from=1, to=0]
        \arrow["\cong"{marking}, draw=none, from=2, to=1-3]
    \end{tikzcd}
\end{equation}
We denote this composite $f \defeq (\fun{P} \circ \fun{U}) \circ q \circ \somega : \fstate \to \fstaterel$.

\begin{proposition}\ifarxiv[proof in~\cref{proof:relative_morphism}]\fi
    \label{prop:relative_morphism}
    The map $f: \fstate \to \fstaterel$ is a colax morphism of CVAs.
\end{proposition}

\Cref{prop:relative_morphism} effectively realises the relative trace model $\fstaterel$ as a quotient of the state trace model $\fstate$.

%% file: content/computation.tex
\section{Local computation}
\label{sec:computation}

Valuation algebras provide a foundation for practical computation through a suite of distributed \emph{local computation} algorithms. These algorithms are designed to resolve \emph{inference problems} that arise in the context of valuation algebras. A comprehensive reference to this topic is~\cite{pouly2012generic}.

\begin{definition}
    Let $\fphi$ be an OVA.\@
    A \defn{knowledgebase} is a finite subset of valuations $K \subseteq \gphi$.
    Let $\cat{A} \defeq \set{A_i \in \cat{T}}_{i \in I}$ be a finite family of domains, so that for each $i \in I$, we have $A_i \subseteq \bigcup_{a \in K} \dd a$.
    Then the task of computing
    $
        \proj{\paren{\bigotimes_{a \in K} a}}{A_i}
    $
    for each $i \in I$, is called the \defn{inference problem for $(K,\cat{A})$}.
    In this context, $\bigotimes_{a \in K} a$ is called the \defn{joint valuation}, and the domains $A_i$ are called \defn{queries}.
\end{definition}

In distributed systems, an inference problem corresponds to determining the local behaviours of a composite system of interacting components. For example, sequential consistency of a specification, as shown in~\cite{Evangelou23}, can be framed as an inference problem. The key to local computation is the combination axiom $\proj{(a \otimes b)}{\dd a} = a \otimes \proj{b}{\dd a \cap \dd b}$. However, traditional theory falls short in our setting as it presumes a single commutative combine operator. Though the generalised combination axiom of~\cref{item:combination} supports local computation for CVAs, further exploration in this area is called for.

%% file: content/discussion.tex
\section{Conclusion}
\label{sec:discussion}
% % i think it means ...

In this work, we have introduced the \emph{concurrent valuation algebra (CVA)}, a new algebraic structure that expands upon \emph{ordered valuation algebras (OVAs)} by incorporating parallel and sequential products. This integration places the theory of concurrent and distributed systems within the expansive scope of valuation algebras.

Our CVAs draw inspiration from existing algebraic frameworks in concurrency theory such as \emph{Communicating Sequential Processes (CSP)}~\cite{DBLP:books/ph/Hoare85}, \emph{Concurrent Kleene Algebra (CKA)}~\cite{DBLP:books/daglib/0008195}, \emph{Concurrent Refinement Algebra (CRA)}~\cite{DBLP:journals/fac/HayesMWC19}, and \emph{duoidal/$2$-monoidal categories}~\cite{MR2724388}. They also facilitate key reasoning methodologies for program specification, like Hoare logic~\cite{DBLP:journals/cacm/Hoare69}, and rely-guarantee reasoning~\cite{jones1981development}.

Within the framework of CVAs, we explored three trace models, each representing distinct computational paradigms, and related them by morphisms.

This research marks a promising pathway to practical applications, particularly through the potent \emph{local computation} framework described in~\cref{sec:computation}. Looking ahead, our work will focus on several key areas. We aim to explore a wider range of CVA models, including the trace semantics of CSP, as well as examples founded on different structures, like trees or transition systems, instead of traces. Our study will further involve deepening the understanding of the general theory of CVAs, including the exploration of their categorical structure, and the ways CVAs on different spaces relate via the pull-back and push-forward mechanisms of their underlying presheaves. Of special interest is the examination of potential links between OVAs and \emph{the monoidal Grothendieck construction}~\cite{moeller2021monoidal}.

%% file: content/ack.tex
\subsubsection{Acknowledgements.}
We convey our sincere gratitude to the following for their valuable insights and support: Alexander Evangelou, Brae Webb, Christina Vasilakopoulou, Cliff Jones, Des FitzGerald, Dylan Braithwaite, Brijesh Dongol, Graeme Smith, Igor Dolinka, James East, Jesse Sigal, Joe Moeller, John Baez, Juerg Kohlas, Kait Lam, Kirsten Winter, Luigi Santocanale, Marc Pouly, Mark Utting, Martti Karvonen, Matt Garcia, Matteo Capucci, Michael Robinson, Mike Shulman, Morgan Rogers, Nick Coughlin, Peter Hoefner, Ralph Sarkis, Reid Barton, Rob Colvin, Scott Heiner, Sori Lee, Ted Goranson, Yannick Chevalier, and the Zulip category theory community.
We are thankful for the support of the Australian Government Research Training Program Scholarship of Naso, and funding from the Australian Research Council (ARC) through the Discovery Grant DP190102142.
We gratefully acknowledge the use of GitHub Copilot and OpenAI ChatGPT software in refining the readability of this paper, though their contribution did not extend to the semantic substance of the research.

%% file: content/appendix.tex
\appendix
\label{app:proofs}

\input{content/app_ova}

\input{content/app_tuple}

\input{content/app_action}

\input{content/app_state}

\input{content/app_relative}

%% file: content/app_ova.tex
% OVA appendix

\section{Proofs of~\cref{sec:ova}}
\label{app:ova}

\subsection{Proof of~\cref{lem:helper}}
\label{proof:helper}

\begin{proof}
    First we show associativity of $\odot$.
    Let $a \in \fphi_A$, $b \in \fphi_B$, and $c \in \fphi_C$, and assume extension-commutation holds.
    We must show that
    \begin{equation*}
        ( a \odot b) \odot c = a \odot ( b \odot c)
    \end{equation*}
    Let $U \defeq A \cup B \cup C$.
    Then
    \begin{align*}
         & ( a \odot b) \odot c                                                                                                                       \\
         & = \ext{(a \odot b)}{U} \odot_U \ext{ c}{U}                                                       &  & \text{(definition of $\odot$)}       \\
         & = \ext{(\ext{ a }{A \cup B} \odot_{A \cup B} \ext{ b}{A \cup B})}{U} \odot_U \ext{ c}{U}         &  & \text{(definition of $\odot$)}       \\
         & = (\ext{( \ext{ a }{A \cup B} ) }{U} \odot_U \ext{(\ext{ b}{A \cup B})}{U} ) \odot_U \ext{ c}{U} &  & \text{(hypothesis: ext.-comm.)}      \\
         & = (\ext{ a }{U} \odot_U \ext{ b}{U} ) \odot_U \ext{ c}{U}                                        &  & \text{(functoriality of extension)}  \\
         & = \ext{ a }{U} \odot_U (\ext{ b}{U} \odot_U \ext{ c}{U})                                         &  & \text{(associativity of $ \odot_U$)} \\
         & = \ext{a}{U} \odot_U (\ext{ (\ext{b}{B \cup C} ) }{U} \odot_U \ext{ (\ext{c}{B \cup C} )}{U} )   &  & \text{(functoriality of extension)}  \\
         & = \ext{a}{U} \odot_U \ext{ ( \ext{b}{B \cup C} \odot_{B \cup C} \ext{c}{B \cup C} ) }{U}         &  & \text{(hypothesis: ext.-comm.)}      \\
         & = \ext{a}{U} \odot_U \ext{( b \odot c )}{U}                                                      &  & \text{(definition of $\odot$)}       \\
         & = a \odot ( b \odot c)                                                                           &  & \text{(definition of $\odot$)}
    \end{align*}
    Thus, $\odot$ is associative.

    To see that $\odot$ is monotone, let $a_1, a_2, b_1, b_2 \in \gphi$ with $a_1 \preceq a_2$ and $b_1 \preceq b_2$.
    Let $U_1 \defeq \dd a_1 \cup \dd b_1$ and $U_2 \defeq \dd a_2 \cup \dd b_2$.
    Noting that $U_2 \subseteq U_1$, then
    \begin{align*}
        a_1 \odot b_1
         & = \ext{a_1}{U_1} \odot_{U_1} \ext{b_1}{U_1}                                            &  & \text{(definition of $\odot$)}      \\
         & \leq_{\fphi_{U_1}} \ext{a_2}{U_1} \odot_{U_1} \ext{b_2}{U_1}                           &  & \text{(local monotonicity)}         \\
         & \leq_{\fphi_{U_1}} \ext{(\ext{a_2}{U_2})}{U_1} \odot_{U_1} \ext{(\ext{b_2}{U_2})}{U_1} &  & \text{(functoriality of extension)} \\
         & \leq_{\fphi_{U_1}} \ext{(\ext{a_2}{U_2} \odot_{U_1} \ext{b_2}{U_2})}{U_1}              &  & \text{(extension-commutation)}      \\
         & \leq_{\fphi_{U_1}} \ext{(a_2 \odot b_2)}{U_1}                                          &  & \text{(definition)}
    \end{align*}
    By the restriction-extension adjunction, this gives us $\proj{(a_1 \odot b_1)}{U_2} \leq_{\fphi_{U_2}} a_2 \odot b_2$, i.e. $a_1 \odot b_1 \preceq a_2 \odot b_2$.
    Hence, the extension $\odot$ is monotone.

    Thus, $(\gphi, \odot)$ is an ordered semigroup.
\end{proof}

\subsection{Proof of~\cref{lem:local_weak_exchange}}
\label{proof:local_weak_exchange}
\begin{proof}
    First, note that extension commutation holds in any OVA $(\fphi, \otimes, \natepsilon)$, as for all $b_1,b_2 \in \fphi_B$ and $B \subseteq A$,
    \begin{equation*}
        \ext{(b_1 \otimes b_2)}{A}
        = \natepsilon_{A} \otimes (b_1 \otimes b_2)
        = (\natepsilon_{A} \otimes \natepsilon_{A}) \otimes (b_1 \otimes b_2)
        = (\natepsilon_{A} \otimes b_1) \otimes (\natepsilon_{A} \otimes b_2)
        = \ext{b_1}{A} \otimes_{A} \ext{b_2}{A}
    \end{equation*}
    Let $a \in \fphi_A$, $b \in \fphi_B$, $c \in \fphi_C$, $d \in \fphi_D$, and $U = A \cup B \cup C \cup D$.
    We have,
    \begin{align*}
         & ( a \parcomp b ) \seq ( c \parcomp d )                                                                                                                                                            \\
         & = \ext{(a \parcomp b)}{U} \seq_U \ext{(c \parcomp d)}{U}                                                                                                   &  & \text{(definition of $\seq$)}     \\
         & = \ext{(\ext{a}{A \cup B} \parcomp_{A \cup B} \ext{b}{A \cup B} )}{U} \seq_U \ext{(\ext{c}{C \cup D} \parcomp_{C \cup D} \ext{d}{C \cup D} )}{U}           &  & \text{(definition of $\parcomp$)} \\
         & = (\ext{(\ext{a}{A \cup B})}{U} \parcomp_{U} \ext{(\ext{b}{A \cup B})}{U}) \seq_U (\ext{(\ext{c}{C \cup D})}{U} \parcomp_{U} \ext{(\ext{d}{C \cup D})}{U}) &  & \text{(ext.-comm.)}               \\
         & = (\ext{a}{U} \parcomp_U \ext{b}{U}) \seq_U (\ext{c}{U} \parcomp_U \ext{d}{U})                                                                             &  & \text{(functoriality)}            \\
         & \leq_{\fphi_U} (\ext{a}{U} \seq_U \ext{c}{U}) \parcomp_U (\ext{b}{U} \seq_U \ext{d}{U})                                                                    &  & \text{(local exchange)}           \\
         & = (\ext{(\ext{a}{A \cup C})}{U} \seq_{U} \ext{(\ext{c}{A \cup C})}{U}) \parcomp_U (\ext{(\ext{b}{B \cup D})}{U} \seq_{U} \ext{(\ext{d}{B \cup D})}{U})     &  & \text{(functoriality)}            \\
         & = \ext{(\ext{a}{A \cup C} \seq_{A \cup C} \ext{c}{A \cup C})}{U} \parcomp_U \ext{(\ext{a}{B \cup D} \seq_{B \cup D} \ext{d}{B \cup D})}{U}                 &  & \text{(ext.-comm.)}               \\
         & = \ext{(a \seq c)}{U} \parcomp_U \ext{(b \seq d )}{U}                                                                                                      &  & \text{(definition of $\seq$)}     \\
         & =   ( a \seq c ) \parcomp ( b \seq d )                                                                                                                     &  & \text{(definition of $\parcomp$)}
    \end{align*}
    The result follows.
\end{proof}

%% file: content/app_tuple.tex
% Tuple systems appendix

\section{Proofs of~\cref{sec:tuple}}
\label{app:tuple}

\subsection{Proof of~\cref{prop:rel_meet}}
\label{proof:meet}
\begin{proof}
    The first claim is trivial, and completeness of $\gpsi$ follows from~\cref{cor:complete_lattice} as $\fpsi$ is actually a presheaf valued in complete lattices.
    For the third, we must show the universal property of meets: that $a \natjoin b$ is the greatest lower bound of $a$ and $b$.
    Clearly $a \natjoin b$ is a lower bound of $a$ and $b$.
    To show it is the greatest, we must show that whenever $(C,c) \preceq (A,a)$ and $(C, c) \preceq (B,b)$, we also have $(C,c) \preceq (A \cup B, a \natjoin b)$.
    So assume $(C,c)$ satisfies the precondition.
    This means that
    \begin{equation*}
        A \subseteq C,\qquad
        B \subseteq C,\qquad
        \proj{c}{A} \subseteq a,\qquad
        \proj{c}{B} \subseteq b
    \end{equation*}
    By the universal property of union, we have $A \cup B \subseteq C$.
    It remains to show that
    \begin{equation*}
        \proj{c}{A \cup B} \subseteq a \natjoin b = \ext{a}{A \cup B} \cap \ext{b}{A \cup B}
    \end{equation*}
    By~\cref{cor:insertion_closure}, functoriality of restriction, and monotonicity of extension, we have
    \begin{equation*}
        \proj{c}{A \cup B}
        \subseteq \ext{(\proj{(\proj{c}{A \cup B})}{A})}{A \cup B}
        = \ext{(\proj{c}{A})}{A \cup B}
        \subseteq \ext{a}{A \cup B}
    \end{equation*}
    and similarly $\proj{c}{A \cup B} \subseteq \ext{b}{A \cup B}$.
    By the universal property of intersection, we have $c \subseteq \ext{a}{A \cup B} \cap \ext{b}{A \cup B}$.
    This shows that $(C,c) \preceq (A \cup B, a \natjoin b)$, as required.
    The result follows.
\end{proof}

\subsection{Proof of~\cref{lem:tuple_lists}}
\label{proof:tuple_lists}
\begin{proof}
    We prove just for $\lists \circ \fomega$, as the proof for $\listsi \circ \fomega$ is similar.
    We have that $\fomega$ is a tuple system, and we must show that $\fun{T} \defeq \lists \circ \fomega$ is a tuple system.
    \begin{description}
        \item[Flasque.] Let $B \subseteq A$ and $t \in \fun{T}_B$.
            We must show that there exists $t' \in \fun{T}_A$ so that $\projtup{t'}{B} = t$.
            Write $t = [t_1, \ldots, t_n]$.
            As $\fomega$ is a tuple system, each $t_i$ has a lifting $t'_i \in \fomega_A$.
            Clearly $t' = [t'_1,\ldots,t'_n]$ is a lifting of $t$.
        \item[Binary gluing.] Let $t_A \in \fun{T}_A$, $t_B \in \fun{T}_B$ be so that
            \begin{equation*}
                s\defeq \projtup{t_A}{A \cap B} = \projtup{t_B}{A \cap B}
            \end{equation*}
            The traces $t_A = [t^A_1,\ldots,t^A_n]$ and $t_B = [t^B_1,\ldots,t^B_n]$ necessarily have the same length and also $\projtup{(t^A_i)}{A \cap B} = \projtup{(t^B_i)}{A \cap B}$ for each $i$.
            As $\fomega$ has the gluing property, we can find a lifting $s_i$ for each pair $t^A_i, t^B_i$ and clearly the trace $s' = [s_1, \ldots, s_n]$ is a common lifting of $t_A$ and $t_B$.
    \end{description}
    The result follows.
\end{proof}

%% file: content/app_action.tex
% Action trace appendix

\section{Proofs of~\cref{sec:action_trace_model}}
\label{app:action}

\subsection{Proof of~\cref{lem:shuf_commute_proj}}
\begin{proof}
    \label{proof:shuf_commute_proj}
    Let $r' \in \proj{(t \shuftrace{A} s)}{B}$.
    Then there exists $r \in t \shuftrace{A} s$ so that $\projtup{r}{B} = r'$.
    Write $t = [r_1, \ldots, r_p]$ and $s = [r_{p+1}, \ldots, r_{p+q}]$.
    By definition of $\shuftrace{A}$, there is $(p,q)$-shuffle $\sigma \in \Sigma_{p,q}$ so that $r = [r_{\sigma(1)}, \ldots, r_{\sigma(p+q)}]$.
    Now $r' = \projtup{r}{B} = [\projtup{r_{\sigma(1)}}{B}, \ldots, \projtup{r_{\sigma(p+q)}}{B}]$ is clearly a shuffle of $\projtup{t}{B} = [\projtup{r_1}{B}, \ldots, \projtup{r_{p}}{B}]$ and $\projtup{s}{B} = [\projtup{r_{p+1}}{B}, \ldots, \projtup{r_{p+q}}{B}]$, thus $r' \in \projtup{t}{B} \shuftrace{B} \projtup{s}{B}$.

    Conversely, let $r \in \projtup{t}{B} \shuftrace{B} \projtup{s}{B}$. Then there exists a $(p,q)$-shuffle $\sigma \in \Sigma_{p,q}$ so that $\projtup{t}{B} = [r_1, \ldots, r_p]$, $\projtup{s}{B} = [r_{p+1}, \ldots, r_{p+q}]$ and $r = [r_{\sigma(1)}, \ldots, r_{\sigma(p+q)}]$.
    By definition of restriction, for each $1 \leq i \leq p + q$, there is a lifting $r'_i$ for $r_i$ so that $t' = [r'_1, \ldots, r'_p]$ is a lifting of $t$, and $s' = [r'_{p+1}, \ldots, r'_{p+q}]$ is a lifting of $s$.
    Now if $r' \defeq [r'_{\sigma(1)}, \ldots, r'_{\sigma(p+q)}]$ then clearly $r' \in t \shuftrace{A} s$, and thus $r = \projtup{r'}{B} \in \proj{(t \shuftrace{A} s)}{B}$.
    The result follows.
\end{proof}

\subsection{Proof of~\cref{prop:conc_ova}}
\label{proof:conc_ova}

\begin{lemma}
    \label{lem:conc_commute_proj}
    For all $t,s \in \atup_A$ and all $B \subseteq A$, we have
    $
        \projtup{(t \conctrace{A} s)}{B} =
        \projtup{t}{B} \conctrace{B} \projtup{s}{B}
    $.
\end{lemma}
\begin{proof}
    \label{proof:conc_commute_proj}
    Let $B \subseteq A$ and $t,s \in \atup_A$, and write $t = [t_1,\ldots,t_n]$, $s = [s_1,\ldots,s_m]$.
    We have
    \begin{align*}
        \projtup{(t \conctrace{A} s)}{B}
         & =\projtup{([t_1,\ldots,t_n] \conctrace{A} [s_1,\ldots,s_m])}{B}                                       \\
         & =\projtup{[t_1,\ldots,t_n, s_1,\ldots,s_m]}{B}                                                        \\
         & =[\projtup{t_1}{B},\ldots,\projtup{t_n}{B}, \projtup{s_1}{B},\ldots,\projtup{s_m}{B}]                 \\
         & =[\projtup{t_1}{B},\ldots,\projtup{t_n}{B}] \conctrace{B} [ \projtup{s_1}{B},\ldots,\projtup{s_m}{B}] \\
         & = \projtup{t}{B} \conctrace{B} \projtup{s}{B}
    \end{align*}
    The result follows.
\end{proof}

\begin{notation}
    \label{not:trace}
    For $A \in \cat{T}$ and a trace $t = [t_1, \ldots, t_n] \in \stup_A$, we write $t_i^j \defeq [t_i, \ldots, t_j]$, where $1 \leq i \leq j \leq n$.
\end{notation}

\begin{lemma}
    \label{lem:conc_ordered_semigroup}
    The structure $(\gstate, \conc)$ is an ordered semigroup.
\end{lemma}
\begin{proof}
    The local operators $\concdom{A}$ are clearly associative.
    By \cref{lem:helper}, it then suffices to show that the local monotonicity and extension-commutation properties hold.
    Let $\lambda : \atup_A \to \Nn$ denote the length function for each $A \in \cat{T}$.
    \begin{description}
        \item[Local monotonicity.] Let $A \in \cat{T}$ and $a_1,a_1',a_2,a_2' \in \fstaterel_A$ with $a_1 \subseteq a_1'$ and $a_2 \subseteq a_2'$.
            Then $a_1 \concdom{A} a_2 \subseteq a_1' \concdom{A} a_2'$ follows from the definition of $\concdom{A}$.
        \item[Extension-commutation.] Let $B \subseteq A$, let $b, b' \in \fstate_B$, let $t \in \ext{( b \concdom{B} b')}{A}$ and let $t' \defeq \projtup{t}{B}$.
            Then there is $r \in b$ and $s \in b'$ with $t' = r \conctrace{B} s$.
            Let $r' \defeq t_1^{\lambda(r)}$ and $s' \defeq t_{\lambda(r) + 1}^{\lambda(t)}$.
            As the length of a trace is preserved by restriction, $\projtup{r'}{B} = r$ and $\projtup{s'}{B} = s$, so that $r' \in \ext{b}{A}$ and $s' \in \ext{b'}{A}$.
            We then have $t = r' \conctrace{A} s' \in \ext{ b }{A} \concdom{A} \ext{ b'}{A}$.

            On the other hand, let $t \in \ext{ b }{A} \concdom{A} \ext{ b'}{A}$.
            Then there exists $r \in \ext{b}{A}$ and $s \in \ext{b'}{A}$ so that $t = r \conctrace{A} s$.
            Now $\projtup{r}{B} \in b$, $\projtup{s}{B} \in b'$, so that $\projtup{r}{B} \conctrace{B} \projtup{s}{B} \in b \concdom{B} b'$.
            As $\projtup{t}{B} = \projtup{r}{B} \conctrace{B} \projtup{s}{B}$ by \cref{lem:conc_commute_proj}, we have that $t \in \ext{( b \concdom{B} b')}{A}$.
    \end{description}
\end{proof}

\begin{lemma}
    \label{lem:conc_comb}
    The operator $\conc$ satisfies the combination axiom.
\end{lemma}
\begin{proof}
    Note that one direction of the combination law follows from monotonicity.
    It then suffices to show suffices to show $a \conc \proj{ b }{A \cap B } \subseteq \proj{ ( a \conc b ) }{A}$ and $\proj{ a }{A \cap B } \conc b \subseteq \proj{ ( a \conc b ) }{B}$.
    Let $ a \in \faction_A$, $b \in \faction_B$, and $t \in a \conc \proj{ b }{A \cap B }$.
    By definition of sequential $\conc$, there exists $t_a \in a$ and $t_b \in \ext{(\proj{b}{A \cap B})}{A}$ so that $t = t_a \conctrace{A} t_b$.
    Let $t_a' \in \ext{a}{A \cup B}$ be a lifting of $t_a$.
    Let
    \begin{equation*}
        s \defeq
        \projtup{t_b}{A \cap B}
        \in
        \proj{(\ext{(\proj{b}{A \cap B})}{A})}{A \cap B}
        = \proj{b}{A \cap B}
    \end{equation*}
    where the equality follows by~\cref{cor:insertion_closure}, and let $s' \in b$ be another lifting of $s$.
    By binary gluing, there then exists a common lifting $t_b' \in \ext{b}{A \cup B}$ of $t_b$ and $s'$.
    Define $t' \defeq t_a' \conctrace{A \cup B} t_b' \in a \conc b$.
    Then by \cref{lem:conc_commute_proj}, we have that
    \begin{equation*}
        t
        = t_a \conctrace{A} t_b
        = \projtup{t_a'}{A} \conctrace{A} \projtup{t_b'}{A}
        = \projtup{(t_a' \conctrace{A \cup B} t_b')}{A}
        = \projtup{t'}{A}
        \in \proj{(a \conc b)}{A}
    \end{equation*}
    Thus,
    $
        a \conc \proj{ b }{A \cap B }
        \subseteq \proj{ ( a \conc b ) }{A}
    $.
    Similarly, $\proj{ a }{A \cap B } \conc b \subseteq \proj{ ( a \conc b ) }{B}$.
    The result follows.
\end{proof}

\begin{proof}[of~\cref{prop:conc_ova}]
    We have that $\conc$ is an ordered semigroup by \cref{lem:conc_ordered_semigroup}.
    The labelling axiom is immediate.
    It is straightforward to see that $\emp$ is a neutral element of $\conc$, and moreover satisfies the strong neutrality condition.
    Finally, $\conc$ satisfies the combination axiom by \cref{lem:conc_comb}.
    The result follows.
\end{proof}

%% file: content/app_state.tex
% State trace appendix

\section{Proofs of~\cref{sec:state_trace_model}}
\label{app:state}

\subsection{Proof of~\cref{prop:glue_ova}}
\label{proof:glue_ova}

\begin{lemma}
    \label{lem:glue_commute_proj}
    For all $t,s \in \stup_A$ and all $B \subseteq A$, we have
    $
        \projtup{(t \gluetrace{A} s)}{B} =
        \projtup{t}{B} \gluetrace{B} \projtup{s}{B}
    $
    and ${(\projtup{t}{B})}^+ = \projtup{(t^+)}{B}$.
\end{lemma}
\begin{proof}
    Write $t = [t_1,\ldots,t_n]$, $s = [s_1,\ldots,s_m]$.
    We have
    \begin{align*}
        \projtup{(t \gluetrace{A} s)}{B}
         & =\projtup{([t_1,\ldots,t_n] \gluetrace{A} [s_1,\ldots,s_m])}{B}                                      \\
         & =\projtup{[t_1,\ldots,t_{n-1}, s_1,\ldots,s_m]}{B}                                                   \\
         & =[\projtup{t_1}{B},\ldots,\projtup{t_{n-1}}{B}, \projtup{s_1}{B},\ldots,\projtup{s_m}{B}]            \\
         & =[\projtup{t_1}{B},\ldots,\projtup{t_n}{B}] \gluetrace{B} [\projtup{s_1}{B},\ldots,\projtup{s_m}{B}] \\
         & = \projtup{t}{B} \gluetrace{B} \projtup{s}{B}
    \end{align*}
    The second claim is immediate.
\end{proof}

\begin{lemma}
    \label{lem:glue_ordered_semigroup}
    The structure $(\fstate, \glue)$ is an ordered semigroup.
\end{lemma}
\begin{proof}
    By \cref{lem:helper}, it suffices to show that the local monotonicity and extension-commutation properties hold.
    \begin{description}
        \item[Local monotonicity.] Let $A \in \cat{T}$ and $a_1,a_1',a_2,a_2' \in \fstate_A$ with $a_1 \subseteq a_1'$ and $a_2 \subseteq a_2'$.
            Then $a_1 \gluedom{A} a_2 \subseteq a_1' \gluedom{A} a_2'$ follows from the definition of $\gluedom{A}$.
        \item[Extension-commutation.] Let $B \subseteq A$, let $b, b' \in \fstate_B$, let $t \in \ext{( b \gluedom{B} b')}{A}$ and let $t' \defeq \projtup{t}{B}$.
            Then there is $r \in b$ and $s \in b'$ with $t' = r \gluetrace{B} s$ and $r^+ = s^-$.
            Using the notation of~\cref{not:trace}, let $r' \defeq t_1^{\lambda(r)}$ and $s' \defeq t_{\lambda(r)}^{\lambda(t)}$.
            As the length of a trace is preserved by restriction, $\projtup{r'}{B} = r$ and $\projtup{s'}{B} = s$, so that $r' \in b$, $s' \in b'$, and also $r'^+ = s'^-$.
            It follows that $t = r' \gluetrace{A} s' \in \ext{ b }{A} \gluedom{A} \ext{ b'}{A}$.

            On the other hand, let $t \in \ext{ b }{A} \gluedom{A} \ext{ b'}{A}$.
            Then there exists $r \in \ext{b}{A}$ and $s \in \ext{b'}{A}$ so that $t = r \gluetrace{A} s$ and $r^+ = s^-$.
            Now $\projtup{r}{B} \in b$, $\projtup{s}{B} \in b'$, and ${(\projtup{r}{B})}^+ = {(\projtup{s}{B})}^-$, so that $\projtup{r}{B} \gluetrace{B} \projtup{s}{B} \in b \gluedom{B} b'$.
            As $\projtup{t}{B} = \projtup{r}{B} \gluetrace{B} \projtup{s}{B}$ by \cref{lem:glue_commute_proj}, we have that $t \in \ext{( b \gluedom{B} b')}{A}$.
    \end{description}
\end{proof}

\begin{lemma}
    \label{lem:glue_comb}
    The operator $\glue$ satisfies the combination axiom.
\end{lemma}
\begin{proof}
    \label{proof:glue_comb}
    Note that one direction of the combination law follows from monotonicity.
    It then suffices to show $a \glue \proj{ b }{A \cap B } \subseteq \proj{ ( a \glue b ) }{A}$ and $\proj{ a }{A \cap B } \glue b \subseteq \proj{ ( a \glue b ) }{B}$.

    Let $ a \in \fstate_A$, $b \in \fstate_B$, and $t \in a \glue \proj{ b }{A \cap B }$.
    By definition of $\glue$, there exists $t_a \in a$ and $t_b \in \ext{(\proj{b}{A \cap B})}{A}$ so that $t = t_a \gluetrace{A} t_b$ and $t_a^+ = t_b^-$.
    Let $t_a' \in \ext{a}{A \cup B}$ be a lifting of $t_a$.
    % By~\cref{lem:extension_int_union}, $\ext{(\proj{b}{A \cap B})}{A} = \proj{(\ext{b}{A \cup B})}{A}$, so let $t_b' \in \ext{b}{A \cup B}$ be a lifting of $t_b$.
    Let
    \begin{equation*}
        s \defeq
        \projtup{t_b}{A \cap B}
        \in
        \proj{(\ext{(\proj{b}{A \cap B})}{A})}{A \cap B}
        = \proj{b}{A \cap B}
    \end{equation*}
    where the equality follows by~\cref{cor:insertion_closure}, and let $s' \in b$ be another lifting of $s$.
    By binary gluing, there then exists a common lifting $t_b' \in \ext{b}{A \cup B}$ of $t_b$ and $s'$.
    We can assume that ${(t_a)}^+ = {(t_b')}^-$; if not, simply replace $t_a'$ by $t_a''$ where $t_a''$ is $t_a'$ with its final component replaced by the first component of $t_b'$.
    Note that $t_a''$ is then still a lifting of $t_a$; for this we must only check its last component restricts onto the last component of $t_a$.
    By~\cref{lem:glue_commute_proj}, we have ${\projtup{(t_a''^+)}{B}} = {\projtup{(t_b'^-)}{B}} = t_b^- = t_a^+$.
    Define $t' \defeq t_a' \gluetrace{A \cup B} t_b' \in a \glue b$.
    Again by \cref{lem:glue_commute_proj}, we have that
    \begin{equation*}
        t
        = t_a \gluetrace{A} t_b
        = \projtup{t_a'}{A} \gluetrace{A} \projtup{t_b'}{A}
        = \projtup{(t_a' \gluetrace{A \cup B} t_b')}{A}
        = \projtup{t'}{A}
        \in \proj{(a \glue b)}{A}
    \end{equation*}
    Thus,
    $
        a \glue \proj{ b }{A \cap B }
        \subseteq \proj{ ( a \glue b ) }{A}
    $.
    Similarly, $\proj{ a }{A \cap B } \glue b \subseteq \proj{ ( a \glue b ) }{B}$.
    The result follows.
\end{proof}

\begin{lemma}
    \label{lem:glue_neutral}
    Gluing product has as neutral element $\nil$, and $\nil$ has the strong neutrality property.
\end{lemma}
\begin{proof}
    Let $a \in \fstate_A$.
    Then
    \begin{align*}
        a \glue \nil_A
         & = \ext{a}{A} \gluedom{A} \ext{\nil_A}{A} \\
         & = a \gluedom{A} \nil_A                   \\
         & = \setc{
            t_a \gluetrace{A} t_\nil
        }{
            t_a \in a ,
            t_\nil \in \nil_A ,
            t_a^+ = t_\nil^-
        }                                           \\
         & = \setc{
            t_a \gluetrace{A} t_\nil
        }{
            t_a \in a, t_\nil \in \stup_A ,
            \lambda(t_\nil) = 1 ,
            t_a^+ = t_\nil^-
        }                                           \\
         & = \setc{
            t_a \gluetrace{A} [x]
        }{
            t_a \in a ,
            x \in \somega_A ,
            t_a^+ = x
        }                                           \\
         & = \setc{
            t_a
        }{
            t_a \in a
        }                                           \\
         & = a,
    \end{align*}
    and similarly $\nil_A \glue a = a$.
    Thus, $\nil$ is a neutral element for gluing product.
    For $A, B \in \cat{T}$, with $B \subseteq A$, we have
    \begin{align*}
        \ext{\nil_B}{A}
         & =\nil_A \glue \nil_B                 \\
         & = \nil_A \gluedom{A} \ext{\nil_B}{A} \\
         & = \setc{
            t_1 \gluetrace{A} t_2
        }{
            t_1 \in \nil_A ,
            t_2 \in \ext{\nil_B}{A} ,
            t_1^+ = t_2^-
        }                                       \\
         & = \setc{
            t_1 \gluetrace{A} t_2
        }{
            t_1 \in \fstate_{A} ,
            t_2 \in \fstate_{A} ,
            t_1^+ = t_2^- ,
            \lambda(t_1) = 1 =
            \lambda(\proj{t_2}{B})
        }                                       \\
         & = \setc{
            [x] \gluetrace{A} [y]
        }{
            x \in \somega_{A} ,
            y \in \somega_{A} ,
            x = y
        }                                       \\
         & = \setc{
            [x]
        }{
            x \in \somega_{A}
        }                                       \\
         & = \nil_{A}
    \end{align*}
    Above, we used the fact that restriction of a trace does not change its length.
    Thus, $\nil$ has the strong neutrality property.
\end{proof}

\begin{proof}[proof of \cref{prop:glue_ova}]
    The ordered semigroup axiom was verified in~\cref{lem:glue_ordered_semigroup}.
    The labelling axiom is immediate from definitions.
    We have shown $\nil$ is a neutral element for $\glue$ that has the strongly neutral property in~\cref{lem:glue_neutral}.
    The combination axiom is proved in~\cref{lem:glue_comb}.
    The result follows.
\end{proof}

%% file: content/app_relative.tex
% Relative state trace appendix

\section{Proofs of~\cref{sec:relative_state_model}}
\label{app:relative}

\subsection{Proof of~\cref{prop:relative_glue_ova}}
\label{proof:relative_glue_ova}

\begin{lemma}
    \label{lem:relative_glue_commute_proj}
    For all $t,s \in \rtup_A$ and all $B \subseteq A$, we have
    $
        \projtup{(t \rgluetrace{A} s)}{B} =
        \projtup{t}{B} \rgluetrace{A} \projtup{s}{B}
    $,
    and ${(\projtup{t}{B})}^+ = \projtup{(t^+)}{B}$.
\end{lemma}
\begin{proof}
    The first claim is due to restriction being a semigroup homomorphism (recall~\cref{fn:forget}).
    The second is clear.
\end{proof}

\begin{lemma}
    \label{lem:relative_glue_ordered_semigroup}
    The structure $(\gstaterel, \rglue)$ is an ordered semigroup.
\end{lemma}
\begin{proof}
    By \cref{lem:helper},
    it suffices to show that the local monotonicity and extension-commutation properties hold.

    \begin{description}
        \item[Local monotonicity.] Let $A \in \cat{T}$ and $a_1,a_1',a_2,a_2' \in \fstaterel_A$ with $a_1 \subseteq a_1'$ and $a_2 \subseteq a_2'$.
            Now $a_1 \rgluedom{A} a_2 \subseteq a_1' \rgluedom{A} a_2'$ follows from the definition of $\gluedom{A}$.
        \item[Extension-commutation.] Let $B \subseteq A$, let $b, b' \in \fstaterel_B$, let
            $t \in \ext{( b \rgluedom{B} b')}{A}$ and let $t' \defeq \projtup{t}{B}$.
            Then there is $r \in b$ and $s \in b'$ with $t' = r \rgluetrace{B} s$ and $r^+ = s^-$.
            Using the notation of~\cref{not:trace}, let $n \in \NnOne$ be so that $\projtup{(t_1^n)}{B} = r$ and $\proj{(t_n^{\lambda(t)})}{B} = s$ (note that $n$ may not be unique, and this is the only point of difference with the proof of~\cref{lem:glue_ordered_semigroup}).
            Let $r' \defeq t_1^{n}$ and $s' \defeq t_{n}^{\lambda(t)}$.
            Then $r' \in \ext{b}{A}$, $s' \in \ext{b'}{A}$, and $r'^+ = s'^-$.
            It follows that $t = r' \rgluetrace{A} s' \in \ext{ b }{A} \rgluedom{A} \ext{ b'}{A}$.

            On the other hand, let $t \in \ext{ b }{A} \rgluedom{A} \ext{ b'}{A}$.
            Then there exists $r \in \ext{b}{A}$ and $s \in \ext{b'}{A}$ so that $t = r \rgluetrace{A} s$ and $r^+ = s^-$.
            Now $\projtup{r}{B} \in b$ and $\projtup{s}{B} \in b'$, and by \cref{lem:relative_glue_commute_proj}, ${(\projtup{r}{B})}^+ = {(\projtup{s}{B})}^-$, and so $\projtup{t}{B} = \projtup{r}{B} \rgluetrace{B} \projtup{s}{B} \in b \rgluedom{B} b'$.
            Thus, $t \in \ext{( b \rgluedom{B} b')}{A}$.
    \end{description}
\end{proof}

\begin{proof}[of~\cref{prop:relative_glue_ova}]
    The ordered semigroup axiom was shown to hold in~\cref{lem:relative_glue_ordered_semigroup}.
    The labelling axiom is immediate.
    Proofs for the neutrality and combination axioms go through exactly as in the proof of \cref{prop:glue_ova}.
    The result follows.
\end{proof}

\subsection{Proof of~\cref{prop:relative_morphism}}
\label{proof:relative_morphism}

\begin{lemma}
    \label{lem:morphism_preserves_ext}
    For all $B \subseteq A$ in $\cat{T}$ and $b \in \fstate_B$, we have
    \begin{equation*}
        f_A(\ext{b}{A}) \subseteq \ext{f_B(b)}{A}
    \end{equation*}
\end{lemma}
\begin{proof}
    We have
    \begin{align*}
        f_B(b)                               & \subseteq f_B(b)          &  & \text{(reflexivity)}                  \\
        \implies f_B(\proj{(\ext{b}{A})}{B}) & \subseteq f_B(b)          &  & \text{(\cref{cor:insertion_closure})} \\
        \implies \proj{f_A(\ext{b}{A})}{B}   & \subseteq f_B(b)          &  & \text{(naturality of $f$)}            \\
        \implies f_A(\ext{b}{A})             & \subseteq \ext{f_B(b)}{A} &  & \text{(adjunction)}
    \end{align*}
    The result follows.
\end{proof}

\begin{lemma}
    \label{lem:morphism_gluetrace}
    Let $t, s \in \stup_A$ so that $t^+ = s^-$.
    Then $q_A(t \gluetrace{A} s) = q_A(t) \rgluetrace{A} q_A(s)$ and ${q_A(t)}^+$ = ${q_A(s)}^-$.
\end{lemma}
\begin{proof}
    The action of $q$ is to eliminate duplicate adjacent components, so the first claim is immediate by observing that we cannot have $t_{\lambda(t)-1} = s^-$ or $t^+ = s_2$ ($t^+ = t_{\lambda(t)} = s^-$).
    For the second, note that $q$ cannot change the first or last components of a trace.
\end{proof}

\begin{proof}[of~\cref{prop:relative_morphism}]
    Let $a \in \fstaterel_A$ and $b \in \fstaterel_B$ for some $A, B \in \cat{T}$.
    \begin{description}
        \item[Colax naturality.] Naturality of $q$ is clear, and this directly implies (strict) naturality of $f$.
        \item[Monotonicity.] For monotonicity, first note that each $f_A$ is monotone as $f_A = \fun{P} ( \fun{U}(q_{\somega_A}))$ and $\fun{P}$ is a functor valued in posets.
            Now if $a \preceq b$ then $B \subseteq A$ and $\proj{a}{B} \subseteq b$.
            Note we have $\dd(f(a)) = A$ and $\dd(f(b)) = B$.
            Then by naturality and local monotonicity, $\proj{f_A(a)}{B} = f_B(\proj{a}{B}) \subseteq f_B(b)$, thus by definition $f(a) \preceq f(b)$.
        \item[Colax unitality.] Clearly, we have in fact $f(\ski) = \rski$ and $f(\nil) = \rnil$.
        \item[Colax multiplicativity.]
            First we show colaxity with respect to $\rnatjoin$.
            We have,
            \begin{align*}
                f(a \natjoin b)
                 & =    f_{A \cup B}(\ext{a}{A \cup B} \cap \ext{b}{A \cup B})                    &  & \text{(definition)}                                                \\
                 & \subseteq f_{A \cup B}(\ext{a}{A \cup B}) \cap f_{A \cup B}(\ext{b}{A \cup B}) &  & \text{(property of image)}                                         \\
                 & \subseteq \ext{f_A(a)}{A \cup B} \cap \ext{f_B(b)}{A \cup B}                   &  & \text{(\cref{lem:morphism_preserves_ext}, monotonicity of $\cap$)} \\
                 & = f(a) \rnatjoin f(b)                                                          &  & \text{(definition)}
            \end{align*}

            Finally, we show colaxity with respect to $\rglue$, i.e.~$f(a \glue b) \preceq f(a) \rglue f(b)$.
            Let
            \begin{equation*}
                t \in f(a \glue b) = f_{A \cup B}(\ext{a}{A \cup B} \glue_{A \cup B} \ext{b}{A \cup B})
            \end{equation*}
            Then there is $r \in \ext{a}{A \cup B} $ and $s \in \ext{b}{A \cup B}$ with $r^+ = s^-$ so that $t = q_{A \cup B}(r \gluetrace{A \cup B} s) $.
            By~\cref{lem:morphism_gluetrace}, $t = q_{A \cup B}(r) \rgluetrace{A \cup B} q_{A \cup B}(s)$, and ${q_{A \cup B}(r)}^+ = {q_{A \cup B}(s)}^-$.
            Note $q_{A \cup B}(r) \in f_{A \cup B}(\ext{a}{A \cup B})$ and $q_{A \cup B}(s) \in f_{A \cup B}(\ext{b}{A \cup B})$,
            so $t$ is a trace in $f_{A \cup B}(\ext{a}{A \cup B}) \rgluedom{A \cup B} f_{A \cup B}(\ext{b}{A \cup B})$.
            By~\cref{lem:morphism_preserves_ext},
            \begin{align*}
                f_{A \cup B}(\ext{a}{A \cup B}) \rglue_{A \cup B} f_{A \cup B}(\ext{b}{A \cup B})
                 & \subseteq\ext{f_A(a)}{A \cup B} \rglue_{A \cup B} \ext{f_B(b)}{A \cup B} \\
                 & = f(a) \rglue f(b)
            \end{align*}
            Thus $f(a \glue b) \preceq f(a) \rglue f(b)$.
    \end{description}
    The result follows.
\end{proof}